\documentclass[a4paper,11pt]{amsart}

\usepackage[top=2cm, bottom=2cm, left=3.5cm, right=3.5cm]{geometry}
\usepackage[colorlinks=true,linkcolor=blue,citecolor=blue,plainpages=false,pdfpagelabels]{hyperref}
\usepackage{amsmath,amssymb,amsthm,amsfonts}
\usepackage{mathtools}
\usepackage{nicefrac}
\usepackage[mathscr]{euscript}
\usepackage{microtype}
\usepackage[sc,osf]{mathpazo}

\makeatletter
\newcommand*\rel@kern[1]{\kern#1\dimexpr\macc@kerna}
\newcommand*\widebar[1]{
  \begingroup
  \def\mathaccent##1##2{
    \rel@kern{0.8}
    \overline{\rel@kern{-0.8}\macc@nucleus\rel@kern{0.2}}
    \rel@kern{-0.2}
  }
  \macc@depth\@ne
  \let\math@bgroup\@empty \let\math@egroup\macc@set@skewchar
  \mathsurround\z@ \frozen@everymath{\mathgroup\macc@group\relax}
  \macc@set@skewchar\relax
  \let\mathaccentV\macc@nested@a
  \macc@nested@a\relax111{#1}
  \endgroup
}
\makeatother

\DeclareMathOperator{\tr}{tr}

\newcommand{\cM}{\mathscr{M}}
\newcommand{\cN}{\mathscr{N}}
\newcommand{\cH}{\mathscr{H}}
\newcommand{\cK}{\mathscr{K}}
\newcommand{\cG}{\mathscr{G}}
\newcommand{\cB}{\mathscr{B}}
\newcommand{\cS}{\mathscr{S}}
\newcommand{\cP}{\mathscr{P}}
\newcommand{\cE}{\mathscr{E}}
\newcommand{\id}{1}

\renewcommand{\vec}[1]{\mathbold{#1}}
\newcommand{\spat}[2]{\Delta(\vec{#1}\!/\!#2)}
\newcommand{\spatn}[2]{\Delta(#1\!/\!#2)}
\newcommand{\transp}{^{\scriptscriptstyle T}}

\newtheorem{lemma}{Lemma}
\newtheorem{proposition}[lemma]{Proposition}
\newtheorem{theorem}[lemma]{Theorem}
\newtheorem{corollary}[lemma]{Corollary}
\theoremstyle{plain}
\newtheorem{definition}[lemma]{Definition}

\usepackage{mdframed}
\newenvironment{finite}
{\vspace{0.2cm}\begin{mdframed}\footnotesize}
{\end{mdframed}\vspace{0.2cm}}

\def\id{{\rm id}}                            
\def\Rl{{\mathbb R}}\def\C{{\mathbb C}}     
\def\norm #1{\Vert #1\Vert}

\def\tr{{\rm Tr}}
\newcommand\Scp[2]{\ensuremath{\langle #1 \vert #2 \rangle}}
\newcommand*{\half}{\frac{1}{2}}
\newcommand*{\vphi}{\varphi}

\newcommand{\dom}{\mathcal{D}}

\begin{document}

\title{R\'enyi divergences\\ as weighted non-commutative vector valued $L_p$-spaces}

\author{Mario Berta}
\address{Department of Computing, Imperial College London, London, United Kingdom}
\email{m.berta@imperial.ac.uk}

\author{Volkher B. Scholz}
\address{Department of Physics, Ghent University, Ghent, Belgium}
\email{volkher.scholz@gmail.com}

\author{Marco Tomamichel}
\address{Centre for Quantum Software and Information, University of Technology Sydney, Sydney, Australia}
\email{marcotom.ch@gmail.com}


\begin{abstract}
  We show that Araki and Masuda's weighted non-commutative vector valued $L_p$-spaces [Araki \& Masuda, Publ. Res. Inst. Math. Sci., 18:339 (1982)] correspond to an algebraic generalization of the sandwiched R\'enyi divergences with parameter $\alpha = \frac{p}{2}$. Using complex interpolation theory, we prove various fundamental properties of these divergences in the setup of von Neumann algebras, including a data-processing inequality and monotonicity in $\alpha$. We thereby also give new proofs for the corresponding finite-dimensional properties. We discuss the limiting cases $\alpha\to \{\frac{1}{2},1,\infty\}$ leading to minus the logarithm of Uhlmann's fidelity, Umegaki's relative entropy, and the max-relative entropy, respectively. As a contribution that might be of independent interest, we derive a Riesz-Thorin theorem for Araki-Masuda $L_p$-spaces and an Araki-Lieb-Thirring inequality for states on von Neumann algebras. 
\end{abstract}

\maketitle


\section{Introduction}

Information-theoretic concepts are often developed under the assumption that the underlying physical systems are described by finite-dimensional Hilbert spaces. It is, however, of fundamental interest to understand which finite-dimensional concepts generalize to physical systems modeled by von Neumann algebras, e.g.\ certain quantum field theories. Moreover, the translation of finite-dimensional arguments to a more abstract theory often requires new proof ideas or at least a streamlining of the original arguments---thereby leading us to a better understanding of the finite-dimensional special case as well.

In this work we explore a generalization to the setup of von Neumann algebras of a family of divergences, the sandwiched R\'enyi divergences~\cite{lennert13,wilde13} (see also~\cite[Sec.~3.3]{jaksic10}), which have found operational meaning and applications in quantum information theory (see, e.g.,~\cite{wilde13,mosonyiogawa13,mosonyi14-2,tomamichelww14,mybook}). A connection between sandwiched R\'enyi divergences and weighted non-commutative $L_p$-spaces has been exploited already in the defining works~\cite{mytutorial12,lennert13,wilde13} and by Beigi~\cite{beigi13} for finite-dimensional systems. In generalizing sandwiched R\'enyi divergences to the algebraic setting it is thus immediate to look for a connection between those and non-commutative $L_p$ spaces defined on von Neumann algebras. This is of course complicated by the fact that when dealing with von Neumann algebras the existence of a trace is not guaranteed and states are not represented by density operators. Nevertheless, quite some progress has been achieved in the recent decades on defining and studying non-commutative $L_p$ spaces defined with respect to von Neumann algebras (see, e.g.,~\cite{pisier03} for an overview). The most well-known such theory are Haagerup's non-commutative $L_p$-spaces~\cite{haagerup79}, as well as their weighted versions (see, e.g.,~\cite{haagerup10}). These are based on suitably extending the algebra such that a trace becomes available again. Another generalization of the commutative theory to von Neumann algebras was given by Araki and Masuda~\cite{araki82}, who define $p$-norms on state vectors (corresponding to purifications in the finite-dimensional case) in the underlying Hilbert space $\cH$, not directly on the state space.

We find that the latter offers a natural generalization of the sandwiched R\'enyi divergences with parameter $\alpha = \frac{p}{2}$, and we hereafter call this generalization Araki-Masuda divergences. Extending their work, we show various properties of this divergence, including a data-processing inequality and monotonicity in $\alpha$. We will also discuss their relation with other divergences such as Umegaki's~\cite{umegaki62} relative entropy (which was generalized to von Neumann algebras by Araki~\cite{araki76}) and Petz' generalization of R\'enyi divergences~\cite{petz85,petz86} (see also~\cite{ohya93} for an overview). Finally, we conjecture that the Araki-Masuda divergences characterize the strong converse exponent of binary hypothesis testing on von Neumann algebras, complementing recent work by Jak\u{s}i\'c \emph{et al.}~\cite{jaksic10}. We hope that our work will lead to further generalizations of information theoretic results to von Neumann algebras such as in~\cite{furrer11}.

Shortly after our manuscript first appeared online, Jen\u{c}ov\'a~\cite{jencova16} proposed another extension of the sandwiched R\'enyi entropies to von Neumann algebras, using instead Kosaki's definition of non-commutative $L_p$-spaces~\cite{Kosaki_1984}. These are in turn based on the already mentioned Haagerup non-commutative $L_p$-spaces, and defined with respect to the state space of the von Neumann algebra in question and not for elements of the underlying Hilbert space. This basic difference implies that the two approaches are complementary to each other, and hence it depends on the exact problem studied which approach is more suitable. For example, Jen\u{c}ov\'a was able to prove the data-processing inequality for positive maps (and not only for completely positive maps as in our setting). However, since Jen\u{c}ov\'a's approach is restricted to the state space, her definition was restricted to values of $\alpha$ larger than one while our approach gives rise to a natural definition for all $\alpha\geq\frac{1}{2}$. In fact, in a follow-up paper Jen\u{c}ov\'a was able to show that her definition is equivalent to our definition for $\alpha>1$ and also obtained an expression for values $\alpha\in\left(\frac{1}{2},1\right)$ in her setting~\cite{jencova17}. From a broader perspective, Jen\u{c}ov\'a's approach is mathematically appealing as it is based on a well established theory. In contrast, Araki and Masuda's definition that we use is less well studied, but offers the advantage that it is defined directly in terms of objects from modular theory. In certain physical applications, such as algebraic quantum field theory, these objects have very explicit expressions (i.e.~as given by the Bisognano-Wichmann theorem~\cite{Bisognano_1975}) and may thus be more suitable for calculations. 


\subsubsection*{Outline} This paper is organized as follows. In Section~\ref{sec:algebraic} we introduce our notation and discuss a few algebraic concepts that we need in the following. In Section~\ref{sec:spaces} we slightly generalize the definition of Araki and Masuda's weighted non-commutative vector valued $L_p$-spaces and discuss some properties. In particular, we prove a Riesz-Thorin theorem for Araki-Masuda $L_p$-spaces. In Section~\ref{sec:divergences} we define the Araki-Masuda divergences, which correspond to an algebraic generalization of the sandwiched R\'enyi divergences. We then discuss various properties and show in particular an Araki-Lieb-Thirring inequality for states on von Neumann algebras. We conclude with Section~\ref{sec:strong_converse} where we discuss our conjecture for an operational interpretation.


\section{Algebraic setup}\label{sec:algebraic}

\subsection{Notation}

We tried to optimize our notation so that the manuscript is most accessible to the quantum information community, thereby sometimes disregarding established conventions in operator theory. We also comment, in framed boxes, on how to translate certain concepts into the language of finite-dimensional quantum theory (sometimes we comment on the commutative case as well).

In these notes, $\cM$ denotes a $W^*$-algebra, $(\cdot)^{\dag}$ is its involution, and $\id \in \cM$ its multiplicative identity. An element $x \in \cM$ is called positive if $x = a^{\dag} a$ for some $a \in \cM$, and the cone of positive elements of $\cM$ is denoted by $\cM_+$.
We denote by $\cP(\cM)$ the space of normal positive linear functionals on $\cM$, i.e.\ elements of the predual of $\cM$ that map $\cM_+$ onto the positive real axis. The subset $\cS(\cM) := \{ \rho \in \cP(\cM) : \rho(\id) = 1 \}$ contains normal states.

For two functionals $\rho, \sigma \in \cP(\cM)$ we write $\rho \ll \sigma$ and say that $\rho$ is supported on (the support of) $\sigma$ if and only if $\sigma(x) = 0 \implies \rho(x) = 0$ for all $x \in \cM$. Strictly stronger, we say that $\rho$ is dominated by $\sigma$ and write $\rho \lll \sigma$ if and only if there exists a positive constant $C$ such that $\rho(x) \leq C \sigma(x)$ for all $x \in \cM_+$.

Moreover, let $\pi : \cM \to \cB(\cH)$ be a normal $*$-representation of $\cM$ as bounded operators acting on a Hilbert space $\cH$ such that $\pi(\cM) \subset \cB(\cH)$ forms a von Neumann algebra. (We only consider normal $*$-representations in this work.) The Hilbert space $\cH$ has an inner product, $\langle \cdot | \cdot \rangle$, which is taken to be anti-linear in the first variable. We say that a vector $\vec{\omega} \in \cH$ implements a functional $\omega \in \cP(\cM)$ if and only if $\omega(a) = \langle \vec{\omega} | \pi(a) \vec{\omega} \rangle$ for all $a \in \cM$. We will consistently use the same Greek letter to denote functionals and vectors that are related in this way, i.e.\ every vector implicitly defines a corresponding functional with the same symbol. 
For any $\sigma \in \cP(\cM)$ we define its support projector $P_{\sigma} \in \cM$ as the minimal projector satisfying $\sigma(P_{\sigma}) = 1$. We define the set $\cH_{\sigma}^{*} := \{ \vec{\rho} \in \cH : \rho \lll \sigma \}$ and its closure, the subspace $\cH_{\sigma}$, which contains all vectors $\vec{\rho}$ implementing functionals with $\rho \ll \sigma$. 
The Gelfand-Naimark-Segal (GNS) construction provides a $*$-representation $\pi_{\sigma} : \cM \to \cB(\cG_{\sigma})$  as well as an implementation $\vec{\sigma} \in \cG_{\sigma}$ of $\sigma$ that is faithful in $\cG_{\sigma}$. Moreover, we may embed $\cG_{\sigma}$ as $\cH_{\sigma}$ into $\cH$ if the latter space allows for an implementation of $\sigma$. In this case we have $\pi_{\sigma}(x) = \pi(P_{\sigma} x P_{\sigma})$.

\begin{finite}
We will often draw on intuition from the finite-dimensional case. Let $\cM_n$ denote the algebra of $n \times n$ matrices, which we conveniently represent as acting on the first tensor factor of a Hilbert space $\mathbb{C}^n \otimes \mathbb{C}^n$, i.e.\ $\pi(x) = x \otimes \id_n$. The trace $\tr$ is implemented by a vector $\vec{\tau} = \sum_{i=1}^n \vec{e}_{i} \otimes \vec{e}_{i}$ for some orthonormal basis $\{ e_{i} \}_{i=1}^n$ of $\mathbb{C}^n$. We denote the transpose with regards to this basis by $(\cdot)\transp$ so that $(x \otimes \id_n) \vec{\tau} = (\id_n \otimes x\transp) \vec{\tau}$. For any $\omega \in \cP(\cM_n)$ we define first the positive element $D_{\omega} \in \cM_n$ via the relation $\omega(x) = \tr\, D_{\omega} x $ and then a vector $\vec{\omega} = (D_{\omega}^{\nicefrac12} \otimes \id) \vec{\tau}$ implementing $\omega$. The vector $\vec{\omega}$ is commonly called the (standard) purification of the density matrix $D_{\omega}$ when $\omega$ is a state. 
The commutant in this case comprises all matrices acting on the second tensor space.
\end{finite}


\subsection{Spatial derivative and relative modular operator} 

Let $\sigma \in \cP(\cM)$ and let $\pi$ be any *-representation of $\cM$ into $\cB(\cH)$. The GNS construction further provides a *-representation $\pi_{\sigma}$ into $\cB(\cH_{\sigma})$ and a vector $\vec{\sigma} \in \cH_{\sigma}$ implementing $\sigma$. Following~\cite[Ch.~4]{ohya93}, for every $\vec{\rho} \in\cH_{\sigma}^*$, we  define the map $R^{\sigma}(\vec{\rho}) : \cH_{\sigma} \to \cH$~by
\begin{align}\label{eq:defR}
  R^{\sigma}(\vec{\rho}) : \pi_{\sigma}(a)\vec{\sigma} \mapsto \pi(a) \vec{\rho}, \quad a \in \cM \,.
\end{align}
Note that this does not overly specify the map since $\pi_{\sigma}(a)\vec{\sigma} = 0$ implies $\pi(a) \vec{\rho} = 0$ for any $a \in \cM$ when $\rho \lll \sigma$.
This map is bounded, i.e.\ there exists a constant $C$ such that for all $a \in \cM$, 
\begin{align}\label{eq:itsbounded}
\| R^{\sigma}(\vec{\rho}) \pi_{\sigma}(a) \vec{\sigma} \|^2 = \| \pi(a) \vec{\rho} \|^2 = \rho(a^{\dag} a) \leq C\, \sigma(a^{\dag} a) = C \| \pi_{\sigma}(a) \vec{\sigma} \|^2 \,,
\end{align}
where the norm is the one induced by the scalar product on $\cH$.
Furthermore, it is easy to verify that $\pi(b) R^{\sigma}(\vec{\rho}) = R^{\sigma}(\vec{\rho}) \pi_{\sigma}(b)$ for any $b \in \cM$, and from this we can conclude that the operator $R^{\sigma}(\vec{\rho}) R^{\sigma}(\vec{\rho})^{\dag}$ lies in the commutant of $\cM$ in $\cB(\cH)$.\footnote{The commutant $\cM^\prime$ of a von Neumann algebra $\cM \subset \cB(\cH)$ is itself a von Neumann algebra and is given by all elements in $\cB(\cH)$ which commute with $\cM$, $\cM^\prime = \{X \in \cB(\cH)\,:\, [X,Y] = 0 \;\forall\,Y \in \cM \}$.} If there exists an element $\vec{\sigma} \in \cH$ which implements $\sigma$, it follows that the map $R^\sigma(\vec{\sigma})$ is a partial isometry.

For every $\vec{\omega} \in \cH$, the functional $\vec{\rho} \mapsto \langle \vec{\omega} | R^{\sigma}(\vec{\rho}) R^{\sigma}(\vec{\rho})^{\dag} \vec{\omega} \rangle$ constitutes a quadratic form to which we associate a positive self-adjoint operator $\spat{\omega}{\sigma}^\half$ on $\cH$. That is, the spatial derivative is defined on its domain $\dom\big(\spat{\omega}{\sigma}^\half\big) = \cH_{\sigma}^*$ as
\begin{align}
  \langle \vec{\rho} | \spat{\omega}{\sigma}  \vec{\rho} \rangle = \Scp{\spat{\omega}{\sigma}^\half  \vec{\rho}}{\spat{\omega}{\sigma}^\half  \vec{\rho}} = \langle \vec{\omega} | R^{\sigma}(\vec{\rho}) R^{\sigma}(\vec{\rho})^{\dag} \vec{\omega} \rangle, \quad \vec{\rho} \in \cH_{\sigma}^* \,,
\end{align}
and thus it is also defined on $\cH_{\sigma}$ (since $\cH_{\sigma}^*$ is dense therein). If $\sigma$ is faithful and $\cH = \cH_{\sigma}$ this is called the relative modular operator and its domain is dense in $\cH$. Functions of the spatial derivative are defined via the functional calculus on its domain. In particular, the operator $\spat{\omega}{\sigma}$ is defined as the square of $\spat{\omega}{\sigma}^\half$. We note however, that the domain of powers of $\spat{\omega}{\sigma}$ may be different than the one of $\spat{\omega}{\sigma}$. An example are solely imaginary powers, as the operator $\spat{\omega}{\sigma}^{it}$ is an isometry for $t \in \Rl$ and hence its domain is the whole space $\cH$.

\begin{finite}
  If $\rho \ll \sigma$ then the operator $R^{\sigma}(\vec{\rho})$ is embedded into $\mathbb{C}^n \otimes \mathbb{C}^n$ as $\id \otimes \big( D_{\rho}^{\nicefrac12} \big)\transp \big( D_{\sigma}^{-\nicefrac12} \big)\transp$, where the inverse is taken on the support of $D_{\sigma}$ in case $\sigma$ is not faithful. The spatial derivative is represented as $\spat{\omega}{\sigma} = D_{\omega} \otimes \big(D_{\sigma}^{-1}\big)\transp$. If $\sigma$ is faithful this is the relative modular operator.
\end{finite}

\begin{finite}
  Let us also consider the commutative case. Let $(X,\Sigma)$ be a measurable space and let $\rho$ be a probability measure and $\sigma$ be a positive measure on the $\sigma$-algebra $\Sigma$. Assume that both are absolutely continuous with respect to another positive measure $\tau$ (we can take, e.g., $\tau = \rho + \sigma$). Then $\rho$, $\sigma$ are naturally regarded as positive functionals on the commutative von Neumann algebra $\cM = L^\infty(X,\Sigma,\tau)$, the almost-everywhere bounded functions. Performing the GNS construction with respect to $\tau$ gives rise to the Hilbert space $L^2(X,\Sigma,\tau)$ on which $L^\infty(X,\Sigma,\tau)$ acts by point-wise multiplication. The Radon-Nikodyn derivatives $\frac{d\rho}{d\tau}$ and $\frac{d\sigma}{d\tau}$ of $\rho$ and $\sigma$ with respect to $\tau$ are positive elements in $L^1(X,\Sigma,\tau)$, and hence their square roots are elements of $L^2(X,\Sigma,\tau)$ which we denote with $\vec{\rho}$ and $\vec{\sigma}$. Considering the action of the map $R^\sigma(\vec{\rho})$ in~\eqref{eq:defR} we see that
    \begin{align}\label{eq:classical1}
      R^\sigma(\vec{\rho}) = \vec{\rho}\,\vec{\sigma}^{-1} = \left(\frac{d\rho}{d\tau}\right)^{\half}\,\left(\frac{d\sigma}{d\tau}\right)^{-\half}\,,
    \end{align}
    the inverse being defined on the support of $\sigma$ and put to $0$ otherwise. Correspondingly, we find for the spatial derivative $\spat{\omega}{\sigma}$ for $\vec{\omega} \in L^2(X,\Sigma,\tau)$ the form
    \begin{align}\label{eq:classical2}
      \spat{\omega}{\sigma} = \vec{\omega}^2 \, \vec{\sigma}^{-2}\,,
    \end{align}
    where again powers of elements in $L^2(X,\Sigma,\tau)$ are defined by the point-wise multiplication of functions.
\end{finite}

Our arguments are mostly based on complex interpolation theory applied to products of spatial derivatives. In particular, the monotonicity in $\alpha$, our extension of Araki-Lieb-Thirring inequality to von Neumann algebras as well as our version of the Riesz-Thorin theorem are applications of the following Lemma. Its proof is deferred to Appendix~\ref{app:thelemma}. We define the complex strip $S_1 = \{z \in \C\,:\,0\leq \Re(z) \leq 1\}$.

\begin{lemma}\label{lem:main_interpolation}
  Let $\cH$, $\cK$ be two Hilbert spaces, $\cM \subset \cB(\cH)$, $\cN \subset \cB(\cK)$ two von Neumann algebras and $V :\cH \to \cK$ a bounded mapping. Consider two affine holomorphic functions $g(z) = g_1 z + g_0$, $h(z) = h_1 z + h_0$, with $g_0,g_1,h_0,h_1 \in \Rl$ and vectors $\vec{\vphi}, \vec{\rho} \in \cH$, $\vec{\omega} \in \cK$ and  $\sigma \in \cP(\cM)$, $\tau \in \cP(\cN)$ such that for $x \in \{0,1\}$ the statement $\vec{\vphi} \in \dom(\spat{\rho}{\sigma}^{g(x)})$ holds. If the vector-valued function
  \begin{align}
    f: S_1 \to \cK,\quad z \mapsto \spat{\omega}{\tau}^{h(z)}\,V\,\spat{\rho}{\sigma}^{g(z)}\vec{\vphi}
  \end{align}
  satisfies $C_x := \sup_{t \in \mathbb{R}} \norm{f(x+it)} < \infty$ for $x \in \{0,1\}$ (i.e.\ it is uniformly bounded on the boundaries of $S_1$), then $f(z)$ is holomorphic in the interior of $S_1$ and satisfies
  \begin{align}
    \norm{f(\theta)} \leq C_0^{1-\theta}\,C_1^{\theta}\,, \qquad \textrm{for $0 \leq \theta \leq 1$.}
  \end{align}
\end{lemma}


\section{Non-commutative Araki-Masuda $L_p$-norms}\label{sec:spaces}

\subsection{Definition}
 
We first recall the definition of weighted non-commutative vector valued $L_p$-spaces due to Araki and Masuda~\cite{araki82} and extend it to the case of a non-faithful reference state $\sigma$ by means of the spatial derivative. 

\begin{definition}\label{def:pnorms}
  Let $\sigma \in \cP(\cM)$ and let $\pi$ be a $*$-representation of $\cM$ in $\cB(\cH)$. We define the following norms on elements of $\cH$.
   For $2\leq p \leq \infty$ we define $\sigma$-weighted $p$-norm of $\vec{\rho} \in \cH$ by
  \begin{align}\label{eq:pnorm1}
    \| \vec{\rho}\| _{p,\sigma}\ := \sup_{\vec{\omega} \in \cH,\,\|\vec{\omega}\|=1} \left\|\spat{\omega}{\sigma}^{\frac12 -\frac{1}{p}} \vec{\rho} \right\| 
  \end{align}
  if $\rho \ll \sigma$ (which may be infinite) and $+\infty$ otherwise.
  For $1\leq p < 2$ we define
  \begin{align}\label{eq:pnorm2}
    \|\vec{\rho}\|_{p,\sigma}\ := \inf_{\vec{\omega} \in \cH,\,\|\vec{\omega}\|=1,\, \omega^\prime \gg \rho^\prime} \left\|\spat{\omega}{\sigma}^{\frac12-\frac{1}{p}} \vec{\rho} \right\| \,,
  \end{align}
  where $\omega^\prime(a'):=\langle \vec{\omega} | a'  \vec{\omega} \rangle$ for all $a'\in \pi(\cM)'$.
\end{definition}

For a fixed, faithful $\sigma \in \cP(\cM)$ these quantities constitute norms on $\cG_{\sigma}$ as shown in~\cite[Thm.~1]{araki82}. Generally, for $1 \leq p < 2$ the quantity is only a semi-norm.
We have $\|\vec{\rho}\|_{2,\sigma} \leq \|\vec{\rho}\|$ with equality when $\rho \ll \sigma$.
For $p > 2$ the norms are finite when $\rho \lll \sigma$ but they can be infinite in general, even when $\rho \ll \sigma$.

Moreover, we note that since the support projection of the spatial derivative $\spat{\omega}{\sigma}$ is contained in the support of $\pi(P_\sigma)$, we have that $\| \vec{\rho}\| _{p,\sigma} = \| \pi(P_{\sigma}) \vec{\rho}\| _{p,\sigma}$ for $1\leq p \leq \infty$. Hence in estimating these norms we can safely assume that we are working on the Hilbert space $\cH_\sigma$, since elements in its complement are projected out.

\begin{finite}
  Let us take a closer look at these expressions in our standard representation.
 We find
  \begin{align}\label{eq:thisoneworks}
   \Big\|\spat{\omega}{\sigma}^{\frac{1-2}{2p}} \vec{\rho} \Big\|^2
   = \bigg\| D_{\omega}^{\frac12-\frac{1}{p}} D_{\rho}^{\frac12} \otimes \Big(D_{\sigma}^{\frac{1}{p} - \frac12}\Big)\transp  \vec{\tau} \bigg\|^2
   = \tr\, D_{\omega}^{1-\frac{2}{p}} D_{\rho}^{\frac12} D_{\sigma}^{\frac{2}{p} - 1} D_{\rho}^{\frac12}  ,
   \end{align}
   which has been studied in~\cite[Def.~5]{lennert13}. Taking the supremum or infimum over $\omega$, respectively, leads to well-known $\sigma$-weighted norms,
   \begin{align}\label{eq:weightednorm}
     \| \vec{\rho} \| _{p,\sigma}\ = \Big\| D_\sigma^{\frac{1}{p} - \frac12} D_\rho^{\frac12} \Big\|_p 
     = \Bigg( \tr \Big( D_\rho^{\frac12} D_\sigma^{\frac{2}{p} -1} D_\rho^{\frac12} \Big)^{\frac{P}{2}} \Bigg)^{\frac{1}{p}} = \Bigg( \tr \Big( D_\sigma^{\frac{1}{p} - \frac12} D_\rho D_\sigma^{\frac{1}{p} - \frac12} \Big)^{\frac{P}{2}} \Bigg)^{\frac{1}{p}} \,,
   \end{align}
   where $\| \cdot \|_p$ denotes the Schatten norm of order $p$ (see also~\cite{audenaert13,bertawilde14} for related discussions). While we can determine the optimizer in the finite-dimensional case, this is not so easily done algebraically. However the variational formula around the expression~\eqref{eq:thisoneworks} generalizes to the algebraic setting as seen in~\eqref{eq:pnorm1} and~\eqref{eq:pnorm2}, if interpreted as a norm on vectors, i.e.\ as a norm on the square root of the density matrix $D_\rho^{\frac12}$ instead of $D_\rho$ itself.
\end{finite}

\begin{finite}
 In the commutative case, the optimization can be performed without any problems following~\eqref{eq:classical1} and~\eqref{eq:classical2} and we arrive at the following expressions
    \begin{align}
      \norm{\rho}_{p,\sigma} = \left(\int_X \tau(\mathrm{d}x) \left(\frac{d\rho}{d\tau}(x)\right)^{\frac{p}{2}}\,\left(\frac{d\sigma}{d\tau}(x)\right)^{1-\frac{p}{2}}\right)^{\frac{2}{p}}\,,
    \end{align}
    where we assumed that $\rho \ll \sigma$ if $p\geq 2$.
\end{finite}

The following lemma shows that this definition in fact is independent of the choice of $*$-representation. In particular, we can interpret Definition~\ref{def:pnorms} as a norm on positive normal functionals on the $W^*$-algebra instead of a norm on the vectors implementing these functionals.
\begin{lemma}
\label{lem:repindep}
  Let $\rho, \sigma \in \cP(\cM)$, and let $\pi: \cM \to \cB(\cH)$ and $\widetilde{\pi}: \cM \to \cB(\widetilde{\cH})$ be two $*$-representations with vectors $\vec{\rho} \in \cH$ and $\tilde{\vec{\rho}} \in \widetilde{\cH}$ both implementing $\rho$. Then, we have $\|\vec{\rho}\|_{p,\sigma} = \|\tilde{\vec{\rho}}\|_{p,\sigma}$ for all $1 \leq p \leq \infty$.
\end{lemma}

\begin{proof}
  The fact that both $*$-representations allow for a vector implementing $\rho$ means that the mapping  $V: \cH \to \widetilde{\cH}$ defined by
  \begin{align}
    V:  V \pi(a) \vec{\rho} \mapsto \widetilde{\pi}(a) \tilde{\vec{\rho}}, \quad a \in \cM \,,
  \end{align}
  is an isometry satisfying $V\pi(x) = \widetilde{\pi}(x)V$. From~\eqref{eq:defR} follows that $R^{\sigma}(\tilde{\vec{\rho}} ) = R^{\sigma}(V \vec{\rho} ) = V R^{\sigma}(\vec{\rho} )$ and hence we find
  $V^{\dag} \spatn{\tilde{\vec{\omega}}}{\sigma} V = \spatn{V^{\dag}\!\tilde{\vec{\omega}}}{\sigma}$
for any $\tilde{\vec{\omega}} \in \widetilde{\cH}$. For $p \geq 2$ we have
\begin{align}
  \big\| \spatn{\tilde{\vec{\omega}}}{\sigma}^{\frac12-\frac{1}{p}} \tilde{\vec{\rho}} \big\|^2
  &= \big\langle \vec{\rho} \big| V^{\dag} \spatn{\tilde{\vec{\omega}}}{\sigma}^{1-\frac{2}{p}} V \vec{\rho} \big\rangle  \\
  &\leq \bigg\langle \vec{\rho} \bigg| \left( V^{\dag} \spatn{\tilde{\vec{\omega}}}{\sigma} V \right)^{1-\frac{2}{p}} \vec{\rho} \bigg\rangle  =\big\| \spatn{V^{\dag}\!\tilde{\vec{\omega}}}{\sigma}^{\frac12-\frac{1}{p}} \vec{\rho} \big\|^2 \,,
\end{align}
where the sole inequality follows from Jensen's inequality~\cite{hansen03} (see Appendix~\ref{app:hansen} for an extension to unbounded operators in our specific setting) and the operator concavity of $t \mapsto t^{r}$ for $r \in [0,1]$. Taking the suprema over $\tilde{\vec{\omega}}$ yields the inequality $\|\tilde{\vec{\rho}}\|_{p,\sigma} \leq \|\vec{\rho}\|_{p,\sigma}$ and equality follows because $\vec{\rho}$ and $\tilde{\vec{\rho}}$ are interchangeable in the above argument. For $1 \leq p < 2$ a similar argument using the operator convexity of $t \mapsto t^r$ for $r \in [-1, 0)$ yields the desired result.
\end{proof}

This allows us to introduce the notation $\| \rho \|_{p,\sigma}$ for any $\rho \in \cP(\cM)$, which refers to the norm of an arbitrary implementation of $\rho$. If $\rho \ll \sigma$ the GNS space $\cH_{\sigma}$ allows for implementations of both $\rho$ and $\sigma$ and is thus a natural choice.

\begin{finite}
 In the finite-dimensional case Lemma~\ref{lem:repindep} simply reaffirms our freedom to choose a purification. Consider a pure state $\rho$ and the trivial $*$-representation $\widetilde{\pi}: x \mapsto x$ of $\cM_n$ into itself where $\rho$ is implemented by the vector satisfying $\tilde{\vec{\rho}} \tilde{\vec{\rho}}^{\dag} = D_{\rho}$. In this representation we have $R^{\sigma}(\tilde{\vec{\rho}})R^{\sigma}(\tilde{\vec{\rho}})^{\dag} = \id_n\, \rho(D_{\sigma}^{-1})$ and $\spatn{\tilde{\vec{\omega}}}{\sigma} = D_{\sigma}^{-1} \omega(\id_n)$. The optimization turns trivial in this case and the norm evaluates to
 \begin{align}
   \| \rho \|_{p,\sigma} = \| \tilde{\vec{\rho}} \|_{p,\sigma} = \Big\| D_{\sigma}^{\frac{1}{p} - \frac12} \tilde{\vec{\rho}} \Big\| = \rho\Big( D_{\sigma}^{\frac{2}{p} - 1} \Big)^{\frac12},
 \end{align}
 which is in agreement with the expression for $\|\vec{\rho}\|_{p,\sigma}$ in~\eqref{eq:weightednorm} specialized for pure states.
\end{finite}


\subsection{Norm duality}

Araki and Masuda~\cite[Thm.~1]{araki82} show that, for any $\sigma \in \cP(\cM)$ and two H\"older conjugates $p, q \geq 1$ with $p^{-1} + q^{-1} = 1$, the corresponding $L_p$- and $L_q$-norms are dual on $\cG_{\sigma}$, namely
\begin{align}\label{eq:norm_duality}
\| \vec{\rho}\| _{p,\sigma} = \sup \{ |\!\langle \vec{\rho} | \vec{\omega} \rangle\!|  : \vec{\omega} \in \cG_{\sigma},\ \| \vec{\omega} \|_{q,\sigma} \leq 1 \}\,.
\end{align}
This constitutes a H\"older inequality for the inner product: 
\begin{align}\label{eq:hoelder}
  |\!\langle \vec{\rho} | \vec{\omega} \rangle\!| \leq 
  \|\vec{\rho}\| _{p,\sigma} \, \| \vec{\omega} \|_{q,\sigma} \,.
\end{align}
These norm duality statements continues to hold on $\cH_\sigma$, with very minor changes to the original proof. 


\subsection{Norm interpolation and convexity}

Building on Araki and Masuda's techniques, we derive the following inequality relating the $L_p$-norm of $\rho \in \cP(\cM)$ for different values of $p$ (recall that the norm only depends on the state, not on the exact vector implementing the state).

\begin{proposition}\label{lem:simpleinterpolate}
  Let $\rho, \sigma \in \cP(\cM)$ and either $p_0, p_1 \geq 2$ or $1 \leq p_0, p_1\leq 2$ be given. Then, we have
   \begin{align}\label{eq:simpleinterpolate}
     \|{\rho}\|_{p_\theta,\sigma} \leq \|{\rho}\|_{p_0,\sigma}^{1-\theta} \|{\rho}\|_{p_1,\sigma}^{\theta}\,, \qquad \textnormal{for} \quad \frac{1}{p_\theta} = \frac{1-\theta}{p_0} + \frac{\theta}{p_1} \,.
   \end{align}
\end{proposition}

\begin{proof}
We use the natural implementation $\vec{\rho}$ of $\rho$ in $\cH_\rho$. If the right hand-side is infinite the claim holds trivially and hence we assume that $\|{\rho}\|_{p_0,\sigma}<\infty$ as well as $\|{\rho}\|_{p_1,\sigma}<\infty$. Furthermore, we assume without loss of generality that $p_1\geq p_0$. 

Let us first consider the case $p_0, p_1 \geq 2$. The strategy is to apply Lemma~\ref{lem:main_interpolation} with the trivial choice $h(z) = 0$ (which reduces the first spatial derivative to a projector which we can safely ignore) and $V$ the identity map on $\cH$. Hence we set
\begin{align}
  f(z) = \spat{\omega}{\sigma}^{g(z)}\vec{\rho}
\end{align}
for $\omega\in\cH_\rho$ with $g(z)=\frac{1}{2}-\big(\frac{z}{p_1}+\frac{1-z}{p_0}\big)$. We confirm that $\vec{\rho}$ is in the domain of the operator $\spat{\omega}{\sigma}^{g(x)}$ for $x \in \{0,1\}$ by assumption, and estimate
\begin{align}
  \norm{f(x+it)} = \left\|\spat{\omega}{\sigma}^{it\left(\frac{1}{p_1}-\frac{1}{p_0}\right)}\spat{\omega}{\sigma}^{\frac{1}{2}-\frac{1}{p_x}}\vec{\rho}\right\|\leq\left\|\spat{\omega}{\sigma}^{\frac{1}{2}-\frac{1}{p_x}}\vec{\rho}\right\|\,,
\end{align}
since the imaginary power of $\spat{\omega}{\sigma}$ is a partial isometry. We find
\begin{align}
  \norm{\spat{\omega}{\sigma}^{\frac{1}{2}-\frac{1}{p_\theta}}\vec{\rho}} = \norm{f(\theta)} \leq \left\|\spat{\omega}{\sigma}^{\frac{1}{2}-\frac{1}{p_0}}\vec{\rho}\right\|^{1-\theta}\left\|\spat{\omega}{\sigma}^{\frac{1}{2}-\frac{1}{p_1}}\vec{\rho}\right\|^\theta\,. 
\end{align}
Taking the suprema over $\vec{\omega}\in\cH_\rho$ with $\|\vec{\omega}\|=1$ we arrive at the assertion.

For $1 \leq p_0,p_1 \leq 2$ we simply use the norm duality in~\eqref{eq:norm_duality} (note that we can assume $\rho \in \cH_\sigma$) to establish
     \begin{align}
       \| \vec{\rho} \|_{p_{\theta},\sigma} &= \sup \big\{ |\!\langle \vec{\rho} | \vec{\omega} \rangle\!| : \vec{\omega} \in \cH_{\rho},\ \| \vec{\omega} \|_{q_{\theta},\sigma} \leq 1 \big\} \\
       &\leq  \sup \big\{ |\!\langle \vec{\rho} | \vec{\omega} \rangle\!|^{1-\theta} |\!\langle \vec{\rho} | \vec{\omega} \rangle\!|^{\theta} : \vec{\omega} \in \cH_{\rho},\ \|\vec{\omega}\|_{q_0,\sigma}^{1-\theta} \|\vec{\omega}\|_{q_1,\sigma}^{\theta} \leq 1 \big\} \\
       &\leq \sup \big\{ \|\vec{\rho}\|_{p_0,\sigma}^{1-\theta} \|\vec{\omega}\|_{q_0,\sigma}^{1-\theta} \|\vec{\rho}\|_{p_1,\sigma}^{\theta} \|\vec{\omega}\|_{q_1,\sigma}^{\theta} : \vec{\omega} \in \cH_{\rho},\ \|\vec{\omega}\|_{q_0,\sigma}^{1-\theta} \|\vec{\omega}\|_{q_1,\sigma}^{\theta} \leq 1 \big\} \\
       &= \|\vec{\rho}\|_{p_0,\sigma}^{1-\theta}  \|\vec{\rho}\|_{p_1,\sigma}^{\theta} \,.
     \end{align}
  Alternatively the statement can also be shown by a variation of the above argument using interpolation theory.
\end{proof}

The following simple corollary is noteworthy and turns the previous result into a convexity statement.
\begin{corollary}\label{cor:convex}
The map $p \mapsto \log \|\rho\|_{p,\sigma}^p$ is convex on $[1,2]$ and $[2,\infty)$ for any $\rho,\sigma \in \cP(\cM)$.
\end{corollary}

\begin{proof}
  Taking the logarithm of~\eqref{eq:simpleinterpolate} and multiplying with $p_{\theta}$, we find
  \begin{align}
   \log \|\rho\|_{p_{\theta},\sigma}^{p_{\theta}} \leq \frac{p_\theta}{p_0}(1-\theta) \log \|\rho\|_{p_{0},\sigma}^{p_{0}} + \frac{p_\theta}{p_1}\theta \log \|\rho\|_{p_1,\sigma}^{p_{1}} \,.
  \end{align}
  It remains to verify that $\frac{p_\theta}{p_0}(1-\theta) + \frac{p_\theta}{p_1}\theta = 1$ and $\frac{p_\theta}{p_0}(1-\theta) p_0 + \frac{p_\theta}{p_1}\theta p_1 = p_{\theta}$.
\end{proof}


\subsection{Interpolation theory of linear operators}

Here we give a version of the Riesz-Thorin theorem for Araki-Masuda $L_p$-spaces.

\begin{theorem}\label{thm:riesz-thorin}
  Let $\cM \subset \cB(\cH)$ and $\cN \subset \cB(\cK)$ be two von Neumann algebras, and $\sigma \in \cP(\cM)$, $\tau \in \cP(\cN)$ two positive functionals. For $T:\cH\to\cK$ bounded we have
\begin{align}
\|T\|_{p_\theta,\sigma\to q_\theta,\tau}\leq\|T\|^{1-\theta}_{p_0,\sigma\to q_0,\tau}\,\|T\|^\theta_{p_1,\sigma\to q_1,\tau}\,, \ \textrm{ where }\  \|T\|_{p,\sigma\to q,\tau}:=\sup_{\vec{\rho}\in\cH}\frac{\|T\vec{\rho}\|_{q,\tau}}{\|\vec{\rho}\|_{p,\sigma}}
\end{align}
for $\frac{1}{p_\theta}=\frac{1-\theta}{p_0}+\frac{\theta}{p_1}$ with $p_0,p_1\geq2$ as well as $\frac{1}{q_\theta}=\frac{1-\theta}{q_0}+\frac{\theta}{q_1}$ with $q_0,q_1\geq2$.
\end{theorem}

\begin{finite}
  For finite-dimensional systems, a similar statement has been proven by Beigi~\cite[Thm.~4]{beigi13} for $\sigma$-weighted norms. We note, however, that while Beigi's result is formulated for the density operators, the Araki-Masuda $L_p$-spaces are Banach spaces for the square roots of the density operators, or equivalently, they define $L_p$-norms on the purifying vectors.
\end{finite}

\begin{proof}
  If the right hand-side is infinite the claim holds trivially and hence we assume in the following that it is finite. Our strategy is to prove the upper bound
  \begin{align}
    \|\spat{\chi}{\tau}^{\frac{1}{2}-\frac{1}{q_\theta}}T\vec{\rho}\| \leq \|T\|^{1-\theta}_{p_0,\sigma\to q_0,\tau}\,\|T\|^\theta_{p_1,\sigma\to q_1,\tau}\,,
  \end{align}
  for dense set of vectors $\rho$ with definite $L_p$-norm $\norm{\vec{\rho}}_{p_\theta,\sigma} \leq 1$ and an arbitrary $\vec{\chi} \in \cK$ with $\norm{\vec{\chi}}\leq 1$. The assertion would then follow by taking the supremum over such $\vec{\rho}$ and $\vec{\chi}$. Hence we first invoke Lemma~\ref{lem:representation_Lp} and choose $\vec{\rho}$ to be of the form
  \begin{align}\label{eq:subst}
    \vec{\rho} = u\spat{\omega}{\sigma}^{\frac{1}{p_\theta}}\vec{\sigma}\,,
  \end{align}
  for $\vec{\omega}\in\cH$ with $\|\vec{\omega}\|=1$ and $u\in\cM^\prime \subset \cB(\cH)$ with $\|u\|\leq1$. We aim to apply Lemma~\ref{lem:main_interpolation} and hence set
  \begin{align}
    f(z) = \spat{\chi}{\tau}^{\frac{1}{2}-\left(\frac{z}{q_1}+\frac{1-z}{q_0}\right)}T u\spat{\omega}{\sigma}^{\frac{z}{p_1}+\frac{1-z}{p_0}}\vec{\sigma}\,,
  \end{align}
  which corresponds to the choice $g(z)=\frac{z}{p_1}+\frac{1-z}{p_0}$, and $h(z)=\frac{1}{2}-\left(\frac{z}{q_1}+\frac{1-z}{q_0}\right)$. Applying Lemma~\ref{lem:representation_Lp} gives
  \begin{align}\label{eq:bound_from_rep_Lpspaces}
    \left\|u\spat{\omega}{\sigma}^{\frac{1}{p_x}+it\left(\frac{1}{p_1}-\frac{1}{p_0}\right)}\right\|_{p_x,\sigma} \leq \|\vec{\omega}\|^{\frac{2}{p_x}}=1
  \end{align}
  which implies $\vec{\sigma} \in \dom(\spat{\omega}{\sigma}^{g(x)})$ for $x\in \{0,1\}$. Moreover, we estimate
  \begin{align}
    \norm{f(x+it)} &\leq  \left\|\spat{\chi}{\tau}^{\frac{1}{2}-\frac{1}{q_x}}T u \spat{\omega}{\sigma}^{\frac{1}{p_x}+it\left(\frac{1}{p_1}-\frac{1}{p_0}\right)}\vec{\sigma}\right\|\\ 
    &\leq\left\|T u \spat{\omega}{\sigma}^{\frac{1}{p_x}+it\left(\frac{1}{p_1}-\frac{1}{p_0}\right)}\right\|_{q_x,\tau}\\
    &\leq \left\|T\right\|_{p_x,\sigma\to q_x,\tau}\left\|u \spat{\omega}{\sigma}^{\frac{1}{p_x}+it\left(\frac{1}{p_1}-\frac{1}{p_0}\right)}\right\|_{p_x,\sigma}\\
    &\leq \left\|T\right\|_{p_x,\sigma\to q_x,\tau} ,
  \end{align}
  where we used in the first step that the imaginary power of $\spat{\chi}{\tau}$ is a partial isometry, in the second and third step the definitions of the norms $\norm{.}_{q_x,\tau}$ and $\left\|T\right\|_{p_x,\sigma\to q_x,\tau}$, respectively, and the estimate from~\eqref{eq:bound_from_rep_Lpspaces} in the last step. The requirements of Lemma~\ref{lem:main_interpolation} are satisfied and we get the estimate
  \begin{align}
    \left\|\spat{\chi}{\tau}^{\frac{1}{2}-\frac{1}{q_\theta}}Tu\spat{\omega}{\sigma}^{\frac{1}{p_\theta}}\vec{\sigma}\right\| = \norm{f(\theta)} \leq \|T\|^{1-\theta}_{p_0,\sigma\to q_0,\tau}\,\|T\|^\theta_{p_1,\sigma\to q_1,\tau}\,,
  \end{align}
  from which the assertion follows by substituting~\eqref{eq:subst}.
\end{proof}


\section{Non-Commutative R\'enyi divergence}\label{sec:divergences}

\subsection{Definition}

We will use the non-commutative $L_p$-norms to define new relative entropic functionals on states of $\cM$, which turn out to be an algebraic generalization of the sandwiched R\'enyi divergences~\cite{lennert13,wilde13} (see also~\cite[Sec.~3.3]{jaksic10}).

\begin{definition}
  Let $\sigma \in \cP(\cM)$, $\rho \in \cS(\cM)$ and $\alpha \in [\frac12,1) \cup (1, \infty)$. Then we define the \emph{Araki-Masuda divergence} of order $\alpha$ as
  \begin{align}
     D_{\alpha}(\rho\|\sigma) := \frac{1}{\alpha-1} \log Q_{\alpha}(\rho\|\sigma), \quad
     Q_{\alpha}(\rho\|\sigma) := \| \rho \|_{2\alpha,\sigma}^{2\alpha} \,.
  \end{align} 
  The quantity $D_{\infty}(\rho\|\sigma)$ is defined as the corresponding limit.
\end{definition}

Recall that Lemma~\ref{lem:repindep} establishes that we are free to chose any $*$-representation on $\vec{\rho}$ implementing $\rho$.

\begin{finite}
With~\eqref{eq:weightednorm} it is easily seen that in the finite-dimensional case the Araki-Masuda divergences correspond to the sandwiched R\'enyi divergences:
\begin{align}
D_{\alpha}(\rho\|\sigma)=\frac{1}{\alpha-1} \log\tr \Big( D_\sigma^{\frac{1-\alpha}{2\alpha}} D_\rho D_\sigma^{\frac{1-\alpha}{2\alpha}} \Big)^\alpha\,.
\end{align}
\end{finite}

As a first property of this divergence we show that it is continuous and monotone as a function of $\alpha$.

\begin{lemma}\label{lem:alpha_mono}
  Let $\sigma \in\cP(\cM)$, $\rho \in \cS(\cM)$. The function $\alpha \mapsto D_{\alpha}(\rho\|\sigma)$ is continuous and monotonically increasing on $[\frac12, 1) \cup (1,\infty]$.
\end{lemma}

\begin{proof}
  For ease of notation we define $\phi(t) := \log Q_{1+t}(\rho\|\sigma)$.
  Corollary~\ref{cor:convex} implies that $\phi(t)$ is convex on $[-\frac12, 0]$ and $[0,\infty)$, and thus continuous in these intervals' interiors. This implies continuity of $\alpha \mapsto D_{\alpha}(\rho\|\sigma) = \frac{1}{\alpha-1}\phi(\alpha-1)$ on $[\frac12, 1) \cup (1,\infty]$. 
  
  Moreover, we find that $\phi(0) = 0$ since $\| \rho \|_{2,\sigma} = 1$ for any state $\rho \in \cS(\cM)$. A standard argument reveals that $\frac{1}{t}\phi(t)$ is monotonically increasing for any convex function $\phi(t)$ with $\phi(0) = 0$. More precisely, for $-\frac12 \leq \alpha < \beta < 0$, we have
  \begin{align}
  \phi(\beta) \leq \frac{\beta}{\alpha} \phi(\alpha) + \left( 1 - \frac{\beta}{\alpha} \right) \phi(0) = \frac{\beta}{\alpha} \phi(\alpha) \,,
  \end{align}
  which yields $\frac{1}{\beta} \phi(\beta) \geq \frac{1}{\alpha}\phi(\alpha)$, and analogous reasoning applies for $0 < \alpha < \beta$.
\end{proof}


\subsection{Some limiting cases}

The extrema $\alpha = \frac12$ and $\alpha = \infty$ correspond to well-known quantities. 
The limit $\alpha \to 1$ is more delicate and will be covered in the next sub-section.

\begin{lemma}
  Let $\rho \in \cS(\cM)$, $\sigma \in \cP(\cM)$. We have $D_{\frac12}(\rho\|\sigma) = -\log F(\rho,\sigma)$ where
  \begin{align}\label{eq:fidelity}
    F(\rho,\sigma) 
    :=& \sup \big\{ |\langle \vec{\rho} | U \vec{\sigma} \rangle|^2 :  U\!: \cH_{\sigma} \to \cH_{\rho} \textnormal{ with }
      \|U\| \leq 1, \\ 
      \nonumber &\qquad \qquad \qquad \qquad U \pi_{\sigma}(x) = \pi_{\rho}(x) U\ \textnormal{ for all } x \in \cM  \big\}
  \end{align}
  denotes Uhlmann's fidelity~\cite{Uhlmann76} in the form of Alberti~\cite{alberti83}, and
  \begin{align}\label{eq:dinfty}
    D_{\infty}(\rho\|\sigma) &= \log \inf \{ C > 0 : \rho(x) \leq C \sigma(x) \textnormal{ for all } x \in \cM_+ \}  \,,
  \end{align}
  which has also been studied in quantum information~\cite{datta08,jain02}.
\end{lemma}

These relations have essentially already been established in~\cite{araki82}, although in our setup we prefer to drop the assumption that $\sigma$ is faithful.

\begin{proof}
  To show the identity~\eqref{eq:dinfty}, first note that for any $C$ such that $\rho(x) \leq C\, \sigma(x)$ for all $x \in \cM_+$, we have
 \begin{align}
   \| \vec{\rho} \|_{\infty,\sigma}^2 &= \sup_{\vec{\omega} \in \cH,\atop \|\vec{\omega}\|=1} \| \spat{\omega}{\sigma}^{\frac12} \vec{\rho} \|^2 
   = \sup_{\vec{\omega} \in \cH,\atop \|\vec{\omega}\|=1} \big\langle \vec{\omega} \big| R^{\sigma}(\vec{\rho}) R^{\sigma}(\vec{\rho})^{\dag} \vec{\omega} \big\rangle
   = \| R^{\sigma}(\vec{\rho}) \|^2  \,,
 \end{align}
 where the latter norm is induced by the vector norms on $\cH$ and $\cH_{\sigma}$.
 Hence the inequality $\| \vec{\rho} \|_{\infty,\sigma}^2 \leq  C$ follows from~\eqref{eq:itsbounded} and holds for all such $C$. On the other hand, for every $x \in \cM_+$, we have
 \begin{align}
   \rho(x) &= \langle \vec{\rho} | \pi(x) \vec{\rho} \rangle 
   = \big\langle R^{\sigma}(\vec{\rho}) \vec{\sigma} \big| \pi(x) R^{\sigma}(\vec{\rho}) \vec{\sigma} \big\rangle \\
   &= \big\langle \pi_{\sigma}(x)^{\frac12} \vec{\sigma} \big| R^{\sigma}(\vec{\rho})^{\dag} R^{\sigma}(\vec{\rho}) \pi_{\sigma}(x)^{\frac12} \vec{\sigma} \rangle \leq \big\| R^{\sigma}(\vec{\rho}) \big\|^2 \sigma(x) ,
 \end{align}
 which implies the other direction.

 For the identity~\eqref{eq:fidelity}, we note that we have by $L_p$-space duality, \eqref{eq:norm_duality} that
 \begin{align}
   \norm{\vec{\rho}}_{1,\sigma} = \sup\{|\Scp{\vec{\rho}}{\vec{\psi}}|\,:\,\norm{\vec{\psi}}_{\infty,\sigma} \leq 1\}\,.
 \end{align}
 Using now the $p=\infty$ case just considered, we see that $\vec{\psi} \in \cH_\rho$ implies $\norm{R^\sigma(\vec{\psi})}\leq 1$. By definition, we have $R^\sigma(\vec{\psi})\pi_\sigma(x) = \pi_\rho(x) R^\sigma(\vec{\psi})$ for $x \in \cM$ and hence the relation $\norm{\vec{\rho}}_{1,\sigma}\leq\sqrt{F(\rho,\sigma)}$ follows since
 \begin{align}
   \Scp{\vec{\rho}}{\vec{\psi}} = \Scp{\vec{\rho}}{R^\sigma(\vec{\psi})\vec{\sigma}}\,.
 \end{align}
  In order to prove `$\geq$' note that for any $\vec{\omega} \in \cH_{\rho}$, we have
  \begin{align}
    |\langle \vec{\rho} | U \vec{\sigma} \rangle| = \big|\big\langle \spat{\omega}{\sigma}^{-\frac12} \vec{\rho} \big| \spat{\omega}{\sigma}^{\frac12} U \vec{\sigma} \big\rangle\big| \leq \big\| \spat{\omega}{\sigma}^{-\frac12} \vec{\rho}\big\| \, \big\| \spat{\omega}{\sigma}^{\frac12} U \vec{\sigma} \big\|
  \end{align}
  by the Cauchy-Schwartz inequality. The second norm simplifies to 
  \begin{align}
    \big\| \spat{\omega}{\sigma}^{\frac12} U \vec{\sigma} \big\| = \big\langle \vec{\omega} \big| R^{\sigma}(U\vec{\sigma}) R^{\sigma}(U \vec{\sigma})^{\dag} \vec{\omega} \big\rangle =
    \omega( U U^{\dag} ) \leq 1 \,,
  \end{align}
  which yields $|\langle \vec{\rho} | U \vec{\sigma} \rangle| \leq \| \spat{\omega}{\sigma}^{-\frac12} \vec{\rho}\|$. 
  Since this holds for all such $U$ and all $\vec{\omega} \in \cH_{\rho}$ we are done.
 \end{proof}

Evidently $D_{\infty}(\rho\|\sigma)$ is finite if and only if $\rho \lll \sigma$. By monotonicity it follows that all $D_{\alpha}(\rho\|\sigma)$ are finite if $\rho \lll \sigma$.


\subsection{Comparison with Umegaki's and Petz' divergences}

The first entropic functional that was generalized to the operator algebraic setting (again in a work of Araki~\cite{araki76}) was the relative entropy as defined by Umegaki~\cite{umegaki62}. It is also useful to contrast the definition via Araki-Masuda $L_p$-spaces with a generalization of R\'enyi divergence proposed by Petz~\cite{petz85} (see also~\cite[Chap.~7]{ohya93}). 

The following definitions are in terms of an implementing vector $\vec{\rho} \in \cH$ of a state $\rho \in \cS(\cM)$ on a von Neumann algebra $\cM \subset \cB(\cH)$; however, the quantities are independent of the specific choice of $*$-representation~\cite[p.~80]{ohya93}. In the following we thus choose the GNS Hilbert space~$\cG_{\rho}$.

\begin{definition}
Let $\sigma \in \cP(\cM)$ and $\rho \in \cS(\cM)$. Umegaki's relative entropy is defined as~\cite{araki76},
\begin{align}
D(\rho\|\sigma):=\langle\vec{\rho}|\log\spat{\rho}{\sigma}\vec{\rho}\rangle\,, \qquad \textrm{with } \vec{\rho} \in \cG_{\rho} \textrm{ implementing } \rho, 
\end{align}
if $\rho \ll \sigma$ (which may be infinite) and $+\infty$ otherwise. Moreover, for $\alpha \in (0,1)$ we define~\cite{petz85},
\begin{align}\label{eq:petzdef}
  \widebar{Q}_{\alpha}(\rho\|\sigma):=\langle \vec{\rho} | \spat{\rho}{\sigma}^{\alpha -1} \vec{\rho} \rangle\,,  \qquad \textrm{with } \vec{\rho} \in \cG_{\rho} \textrm{ implementing } \rho \,.
\end{align}
For $\alpha \in (1,2)$ we use the definition in~\eqref{eq:petzdef} when $\rho \ll \sigma$ (which may be infinite) and set $\widebar{Q}_{\alpha}(\rho\|\sigma) = +\infty$ otherwise.
Finally, Petz' R\'enyi divergences are given by
\begin{align}
\widebar{D}_{\alpha}(\rho\|\sigma) := \frac{1}{\alpha-1} \log \widebar{Q}_{\alpha}(\rho\|\sigma)\,.
\end{align}
\end{definition}

\begin{finite}
In the finite-dimensional case Umegaki's relative entropy~\cite{umegaki62} is given as
\begin{align}
D(\rho\|\sigma)=\tr\, D_\rho(\log D_\rho-\log D_\sigma)\, ,
\end{align}
and the Petz divergences~\cite{petz86} can be written as
\begin{align}
\widebar{D}_{\alpha}(\rho\|\sigma)=\frac{1}{\alpha-1} \log\tr \Big(D_\rho^\alpha D_\sigma^{1-\alpha}\Big)\,.
\end{align}
\end{finite}

We collect some useful properties of Petz' version of R\'enyi divergences into the following proposition. Their proof is standard and deferred to Appendix~\ref{app:misc_lemmas}. 

\begin{proposition}\label{pr:petzprop}
   Let $\sigma \in \cP(\cM)$ and $\rho \in \cS(\cM)$. Then, we have:
   \begin{enumerate}
     \item The function $\alpha \mapsto \widebar{D}_{\alpha}(\rho\|\sigma)$ is monotonically increasing on $[0, 2]$.
     \item If $\rho \lll \sigma$ then $\widebar{D}_{\alpha}(\rho\|\sigma)$ is finite for all $\alpha \in [0, 2]$.
     \item We have $\lim_{\alpha \nearrow 1} \widebar{D}_{\alpha}(\rho\|\sigma) = D(\rho\|\sigma)$.
     Moreover, if $\widebar{D}_{\alpha}(\rho\|\sigma)$ is finite for any $\alpha \in (1,2)$, we also have
     \begin{align}
       \lim_{\alpha \searrow 1} \widebar{D}_{\alpha}(\rho\|\sigma) = D(\rho\|\sigma) \,.
     \end{align}
   \end{enumerate}
\end{proposition}


\subsubsection{An algebraic Araki-Lieb-Thirring inequality}

The existence of two different divergences generalizing the commutative case obviously begs the question on their relation to each other. In finite dimensions, an ordering is implied by the Araki-Lieb-Thirring (ALT)~\cite{lieb76,araki90} inequality.

\begin{finite}
  In the finite-dimensional case the ALT inequality~\cite{lieb76,araki90} for density matrices implies that
  \begin{align}
     \| \rho \|_{p,\sigma}^p 
     = \tr \Big( D_\rho^{\frac12} D_\sigma^{\frac{2}{p} - 1} D_\rho^{\frac12} \Big)^{\frac{P}{2}} 
     \leq \tr\, D_{\rho}^{\frac{p}{2}} D_{\sigma}^{1- \frac{p}{2}} = \| \spat{\rho}{\sigma}^{\frac{p}{4}} \vec{\sigma} \|^2, \qquad \textnormal{for} \quad p \geq 2 \,.
  \end{align}
  The inequality holds in the opposite direction for $1 \leq p \leq 2$.
\end{finite}

This motivates the following extension of the ALT inequality to the setting of $W^*$-algebras and non-commutative vector valued $L_p$-spaces. We note that Kosaki already established a version of the ALT inequality for von Neumann algebras~\cite{kosaki92}, albeit only for those possessing a semifinite trace. In contrast, our ALT inequality holds for any von Neumann algebra and is formulated in terms of positive functionals.

\begin{theorem}\label{thm:alt}
  Let $\cM$ be a $W^*$-algebra. For $\rho \in \cS(\cM)$ and $\sigma \in \cP(\cM)$ we have that
  \begin{align}
    D_{\alpha}(\rho\|\sigma) \leq \widebar{D}_{\alpha}(\rho\|\sigma)\,.
  \end{align}
    In fact, if $\cM \subset \cB(\cH)$ is a von Neumann algebra such that $\cH$ supports implementations $\vec{\rho}, \vec{\sigma} \in \cH$ of $\rho$ and $\sigma$, respectively, then we have for $p \geq 2$ that
\begin{align}
\| \vec{\rho} \|_{p,\sigma}^p \leq \| \spat{\rho}{\sigma}^{\frac{p}{4}} \vec{\sigma} \|^2\,.
\end{align}
The inequality holds in the opposite direction for $1 \leq p \leq 2$.
\end{theorem}

We prove the second statement for $p \geq 2$ and postpone the case $1 \leq p \leq 2$ to Appendix~\ref{app:misc_lemmas}, due to its more complicated nature. The first statement follows from the independence of the $L_p$-norm concerning the vector representative of the state (Lemma~\ref{lem:repindep}) as well as from the fact that $\vec{\rho} = \spat{\rho}{\sigma}^\half\vec{\sigma}$ if there exists an element $\vec{\sigma} \in \cH$ implementing $\sigma$.

\begin{proof}
 If the right hand-side is infinite the claim holds trivially and hence we assume that $\| \spat{\rho}{\sigma}^{\frac{p}{4}} \vec{\sigma} \|^2<\infty$, from which it follows that $\vec{\sigma}\in \dom\big(\spat{\rho}{\sigma}^{\frac{p}{4}}\big)$. We will prove 
 \begin{align}\label{eq:more_general}
   \big\| \spat{\omega}{\sigma}^{\frac{1}{2} - \frac{1}{p}} \vec{\rho} \big\| \leq  \big\| \spat{\rho}{\sigma}^{\frac{p}{4}} \vec{\sigma} \big\|^{\frac{2}{p}}\,,
 \end{align}
 for $\vec{\omega}\in\cH$ with $\|\vec{\omega} \| = 1$, from which the assertion follows by taking the supremum over all such vectors. Our strategy is again based on Lemma~\ref{lem:main_interpolation} applied to the function
 \begin{align}
   f(z) = \spat{\omega}{\sigma}^{h(z)}\spat{\rho}{\sigma}^{g(z)}\vec{\sigma}
 \end{align}
 with $g(z)=\frac{zp}{4}$, and $h(z)=\frac{1-z}{2}$. It follows by assumption that $\vec{\sigma}\in \dom\big(\spat{\rho}{\sigma}^{g(x)}\big)$ for $x \in \{0,1\}$. Moreover, we can estimate
 \begin{align}
   \norm{f(x+it)} = \left\|\spat{\omega}{\sigma}^{h(x+it)}\spat{\rho}{\sigma}^{g(x+it)}\vec{\sigma}\right\|
\leq\left\|\spat{\omega}{\sigma}^{\frac{1-x}{2}}\spat{\rho}{\sigma}^{\frac{(x+it)p}{4}}\vec{\sigma}\right\|\,,
 \end{align}
 again using that the imaginary power of a spatial derivative is a partial isometry. For $x=0$ we find 
 \begin{align}
\left\|\spat{\omega}{\sigma}^{\frac{1}{2}}\spat{\rho}{\sigma}^{\frac{itp}{4}}\vec{\sigma}\right\|&\leq\left\|\spat{\omega}{\sigma}^{\frac{1}{2}-\frac{itp}{4}}\spat{\rho}{\sigma}^{\frac{itp}{4}}\vec{\sigma}\right\|=\left\|\spat{\rho}{\omega}^{\frac{itp}{4}}\vec{\omega}\right\|\leq 1\,,
 \end{align}
 where we used Lemma~\ref{lem:modular_identity} in the second step as well as $\|\vec{\omega}\|=1$ in the third step. For $x=1$ we get
 \begin{align}
   \norm{f(1+it)} \leq \left\|P_\sigma\spat{\rho}\sigma^{\frac{(1+it)p}{4}}\vec{\sigma}\right\|&\leq\left\|\spat{\rho}{\sigma}^{\frac{p}{4}}\vec{\sigma}\right\|\,.
 \end{align}
 The requirements of Lemma~\ref{lem:main_interpolation} are satisfied and we get
 \begin{align}
   \left\|\spat{\omega}{\sigma}^{\frac{1-\theta}{2}}\spat{\rho}{\sigma}^{\frac{\theta p}{4}}\vec{\sigma}\rangle\right| = \norm{f(\theta)} \leq  \left\|\spat{\rho}{\sigma}^{\frac{p}{4}}\vec{\sigma}\right\|^\theta\,, 
 \end{align}
 for $0\leq\theta\leq1$. Choosing $\theta=\frac{2}{p}$, using that $\spat{\rho}{\sigma}^{\frac{1}{2}}\vec{\sigma}=\vec{\rho}$, and taking the suprema over all $\omega \in \cH_{\rho}$ with $\|\vec{\omega} \| = 1$ implies the claim~\eqref{eq:more_general}.
\end{proof}


\subsubsection{Umegaki's relative entropy as the limit $\alpha \to 1$}

We find that the Araki-Masuda divergence always converges to Umegaki's relative entropy when $\alpha$ approaches $1$  from below and that if $\rho\lll\sigma$ the function $\alpha \to D_{\alpha}(\rho\|\sigma)$ can be continuously extended to the whole range $\alpha \in \big[\frac12,\infty\big]$.

\begin{theorem}
Let $\cM$ be a $W^*$-algebra, $\sigma \in \cP(\cM)$ and $\rho \in \cS(\cM)$. Then, we have
\begin{align}\label{eq:limit1}
\lim_{\alpha\nearrow 1}D_{\alpha}(\rho\|\sigma) = D(\rho\|\sigma)\,.
\end{align}
If furthermore $\rho\lll\sigma$ then we also have
\begin{align}\label{eq:limit2}
\lim_{\alpha \searrow 1}D_{\alpha}(\rho\|\sigma)=D(\rho\|\sigma)\,.
\end{align}
\end{theorem}

Given the above theorem, it makes sense to define the Araki-Masuda divergence of order one as the limit $D_1(\rho\|\sigma) :=\lim_{\alpha\nearrow 1}D_{\alpha}(\rho\|\sigma)$.

\begin{proof}
We start with the limit from below in~\eqref{eq:limit1} and first show the direction `$\geq$'. We may assume that $D_{\alpha}(\rho\|\sigma)<\infty$ for $\alpha\in[\frac12,1)$ as otherwise the statement is trivial.
We estimate
\begin{align}
D_1(\rho\|\sigma)&=\lim_{\alpha \nearrow 1}\frac{1}{\alpha-1}\log\inf_{\vec{\omega} \in \cH,\,\|\vec{\omega}\|=1}\left\|\spat{\omega}{\sigma}^{\frac12-\frac{1}{2\alpha}} \vec{\rho}\right\|^{2\alpha}\label{eq:geq1}\\
&\geq\lim_{\alpha \nearrow 1}\frac{\alpha}{\alpha-1}\log\langle\vec{\rho}|\spat{\rho}{\sigma}^{1-\frac{1}{\alpha}}\vec{\rho}\rangle\\
&=\lim_{\alpha \nearrow 1}\frac{\alpha}{\alpha-1}\log\int_0^{\infty}t^{1-\frac{1}{\alpha}}\langle\vec{\rho}|P(\mathrm{d}t)\vec{\rho}\rangle\,,
\end{align}
with the measure $P(\mathrm{d}t)$ from the spectral decomposition of $\spat{\rho}{\sigma}$. For $\beta:=1-\frac{1}{\alpha}$ and $\mu(\mathrm{d}t):=\langle\vec{\rho}|P(\mathrm{d}t)\vec{\rho}\rangle$ we calculate with de L'Hospital's rule and the dominated convergence theorem that the left limit evaluates to
\begin{align}\label{eq:beta_limit}
\lim_{\beta\nearrow 0}\frac{\log\left(\int_0^{\infty}t^\beta\mu(\mathrm{d}t)\right)}{\beta} =
\lim_{\beta\nearrow 0}\ \frac{\mathrm{d}}{\mathrm{d}\beta} \int_0^{\infty} t^\beta \mu(\mathrm{d}t) =
\int_0^{\infty}\log(t)\mu(\mathrm{d}t)\,,
\end{align}
from which it follows that
\begin{align}\label{eq:geq2}
D_1(\rho\|\sigma)\geq\langle\vec{\rho}|\log\spat{\rho}{\sigma}\vec{\rho}\rangle=D(\rho\|\sigma)\,.
\end{align}

To show the direction `$\leq$' of~\eqref{eq:limit1} we invoke the ALT inequality in Theorem~\ref{thm:alt} which states that
$D_\alpha(\rho\|\sigma)\leq\widebar{D}_{\alpha}(\rho\|\sigma)$.
The claim then follows from Proposition~\ref{pr:petzprop}, which establishes that
$\lim_{\alpha\nearrow 1}\widebar{D}_{\alpha}(\rho\|\sigma)=D(\rho\|\sigma)$.

Let us now proceed to the limit from above in~\eqref{eq:limit2}.
The direction `$\geq$' follows analogously to~\eqref{eq:geq1}--\eqref{eq:geq2}, where by assumption we have that $D_{\alpha}(\rho\|\sigma)$ is finite for all $\alpha > 1$ and calculate the limit $\beta \searrow 0$. 

To get the direction `$\leq$' we again invoke the ALT inequality in Theorem~\ref{thm:alt} and are left to show
$\lim_{\alpha \searrow 1}\widebar{D}_{\alpha}(\rho\|\sigma)=D(\rho\|\sigma)$ which follows from Proposition~\ref{pr:petzprop}, assertions (3), the requirements of which are satisfied by our assumption $\rho \lll \sigma$, c.f. assertion (2) of Proposition~\ref{pr:petzprop}.
\end{proof}


\subsection{The data-processing inequality}

We consider two $W^*$-algebras $\cM$, $\cN$ and completely positive and unital maps $\cE:\cN\to\cM$. Moreover, we assume that these maps are normal, that is, they have a pre-dual $\cE_*$ mapping the set of normal functionals on $\cM$ into the set of normal functionals on $\cN$. We call $\cE$ a quantum channel from $\cN$ to $\cM$. Assuming $\cM \subset \cB(\cH)$ for some Hilbert space $\cH$, there exists by Stinespring's theorem~\cite{stinespring54} a Hilbert space $\cK$, a $*$-representation $\pi$ of $\cN$ on $\cK$, and an isometry $T:\cH\to\cK$ such that $\cE(a)=T^{\dagger} \pi(a) T$ for all $a\in\cN$. We call the triple $\big(\cK,\pi,T\big)$ a Stinespring dilation of $\cE$.

\begin{theorem}\label{thm:data_processing}
Let $\cM$, $\cN$ be two $W^*$-algebras, $\rho\in\cS(\cM)$, $\sigma\in\cP(\cM)$, and $\cE : \cN \to \cM$ be a quantum channel from $\cN$ to $\cM$. Then, we have
\begin{align}
D_\alpha(\rho\|\sigma)\geq D_\alpha(\cE_*(\rho)\|\cE_*(\sigma))\quad\text{for all $\alpha\in \Big[\frac{1}{2},1\Big )\cup(1,\infty]$.}
\end{align}
In fact, assuming $\cM \subset \cB(\cH)$, for any Stinespring dilation $\big(\cK,\pi,T\big)$ of $\cE$ we have for $p\geq 2$ that
\begin{align}\label{eq:dpi-vectors}
\left\|T\vec{\rho}\right\|_{p,\cE_*(\sigma)}\leq\|\vec{\rho}\|_{p,\sigma}\quad\text{for all $\rho\in\cH$.}
\end{align}
The inequality holds in the opposite direction for $1\leq p\leq 2$.
\end{theorem}

The first assertion of Theorem~\ref{thm:data_processing} follows immediately from the second by noting that $T\vec{\rho}$ is an implementing vector of $\cE_*(\rho)$ for all $\rho\in\cS(\cM)$ with $\vec{\rho}$ implementing~$\rho$. In the limit $\alpha\to1$ we get Umegaki's relative entropy and the data-processing inequality holds as well~\cite{lieb73,uhlmann77}. The case $\alpha=\half$ corresponds to the data-processing inequality of the fidelity and is due to Alberti~\cite{alberti83}. 

\begin{finite}
  In the finite-dimensional case Beigi~\cite{beigi13} used a similar argument to prove the data-processing inequality from Riesz-Thorin for $\alpha > 1$. Data-processing in the complete range $\alpha \in [\frac12,1) \cup (1, \infty)$ was shown by Frank and Lieb~\cite{frank13}.
   Recently it was noticed that Beigi's proof continues to hold for (not necessarily completely) positive unital maps~\cite{hermes15} (recall that we state our results for maps defined in the Heisenberg picture). While our proof also works for $\alpha \in \big[\frac12,1\big)$, we rely on the Stinespring dilation in our setup and are thus restricted to completely positive unital maps.
\end{finite}

\begin{proof}[Proof of Theorem~\ref{thm:data_processing}]
First note that the claim in~\eqref{eq:dpi-vectors} can be equivalently formulated as
\begin{align}
   \|T\|_{p,\sigma\to p,\mathcal{E}_*(\sigma)} &\leq 1 \quad \textrm{for $p \geq 2$} \,.
\end{align}
In the limiting case $p\to\infty$ this statement follows trivially from the expression~\eqref{eq:dinfty} for $\| \cdot \|_{\infty,\sigma}$ and the positivity property of the quantum channel. The case $p=2$ follows (with equality) from $\|\cdot\|_{2,\sigma}=\|\cdot\|$ and the fact that $T$ is an isometry. The general case for $p\geq2$ follows by applying Theorem~\ref{thm:riesz-thorin} with the parameter choices $p_0=q_0=\infty$, $p_1=q_1=2$, $\theta=\frac{2}{p}$, $\sigma=\sigma$, and $\tau=\cE_*(\sigma)$, yielding
\begin{align}
\|T\|_{p,\sigma\to p,\tau}\leq\|T\|^{1-\frac{2}{p}}_{\infty,\sigma\to \infty,\cE_*(\sigma)}\,\|T\|_{2,\sigma\to 2,\cE_*(\sigma)}^{\frac{2}{p}} \leq 1 \,.
\end{align}

For the case $1\leq p\leq2$, we simply use the norm duality in~\eqref{eq:norm_duality} to establish
\begin{align}
\|\vec{\rho}\|_{p,\sigma}=\sup\left\{\left|\langle\vec{\psi}\middle|\vec{\rho}\rangle\right|\,:\,\|\vec{\psi}\|_{q,\sigma}\leq1\right\}&=\sup\left\{\left|\langle T\vec{\psi}\middle|T\vec{\rho}\rangle\right|\,:\,\|\vec{\psi}\|_{q,\sigma}\leq1\right\}\\
&\leq\left\|T\vec{\psi}\right\|_{q,\cN_*(\sigma)}\left\|T\vec{\rho}\right\|_{p,\cN_*(\sigma)}\\
&\leq\left\|T\vec{\rho}\right\|_{p,\cN_*(\sigma)}\,,
\end{align}
for $\frac{1}{q}+\frac{1}{p}=1$, and where we used $T^{\dagger}T=\id$ (since $T$ is an isometry) as well as the statement of the lemma for $p\geq2$.
\end{proof}


\section{Application to hypothesis testing}\label{sec:strong_converse}

We have used Araki-Masuda weighted non-commutative vector valued $L_p$-spaces to define an algebraic generalization of the sandwiched R\'enyi divergences. We have shown various properties of these divergences, including a data-processing inequality and monotonicity in the parameter $\alpha = \frac{p}{2}$. We have also shown that the Araki-Masuda divergences are lower bounds on an earlier non-commutative generalization of R\'enyi divergence by Petz in the range $\alpha \in [0, 2]$. The latter quantities attain operational meaning in binary hypothesis testing on von Neumann algebras, as shown by Jak\u{s}i\'c \emph{et al.}~\cite{jaksic10}. Our work elicits the question whether the Araki-Masuda divergences characterize the strong converse exponent in binary hypothesis testing on von Neumann algebras. Mosonyi and Ogawa~\cite{mosonyiogawa13} showed that this is the case in the finite-dimensional setting. 

Let us recapitulate the notation and setup in~\cite{jaksic10} and the result in~\cite{mosonyiogawa13}
for the case of binary hypothesis testing between identical product states. (We refer the reader to these papers for a more comprehensive discussion.) Let $\cM \subset \cB(\cH)$ be a von Neumann algebra and let $\rho, \tau \in \cS(\cM)$ be represented by vectors $\vec{\rho}, \vec{\tau} \in \cH$. A (hypothesis) test $T$ is a positive contraction in $\cM$, i.e.\ $\{T, \id - T \} \subset \cM_+$.
For any test $T$, we define the errors of the first and second kind as $\rho(\id - T)$ and $\tau(T)$, respectively. We use the von Neumann tensor product~\cite[Ch.~4]{takesakibook1}, denoted $\cM \,\widebar{\otimes}\, \cM$, to define a sequence of von Neumann algebras $\cM_n = \cM^{\widebar{\otimes}n} \subset \cB(\cH^{\otimes n})$, and two sequences of states in $\rho_n, \tau_n \in \cS(\cM_n)$ determined by their implementing vectors $\vec{\rho}_n = \vec{\rho}^{\otimes n}$ and $\vec{\tau}_n = \vec{\tau}^{\otimes n}$ in $\cH^{\otimes n}$. 

Now let us consider a sequence of tests $\{ T_n \}$ such that the error of the second kind satisfies $\tau_n(T_n) \leq \exp(-n r)$ for some $r > D(\rho\|\tau)$. The strong converse to Stein's lemma (see, e.g.~\cite{jaksic10}) tells us that in this case the error of the first kind, $\rho_n(\id - T_n)$, converges to $1$ as $n \to \infty$. In fact, in the finite-dimensional case~\cite{ogawa00} it is known that this convergence is exponential in $n$.
 Let us thus define the optimal strong converse exponent as
\begin{align}
 B_e^{*}(r) := \inf_{ \{ T_n \} }  \left\{  \limsup_{n \to \infty} -\frac{1}{n} \log \rho_n(T_n) \big) \,\middle|\, \liminf_{n \to \infty} -\frac{1}{n} \log \tau_n(T_n) \geq r \right\}
\end{align}
Mosoyniy and Ogawa~\cite{mosonyiogawa13}, again in the finite-dimensional case, show that
\begin{align}
  B_e^{*}(r) = \sup_{\alpha > 1} \frac{\alpha-1}{\alpha} \left( r - D_{\alpha}(\rho\|\tau) \right),
\end{align}
yielding an operational interpretation of the sandwiched R\'enyi divergence for $\alpha \geq 1$. We conjecture that this relation also holds in the general algebraic case with the sandwiched R\'enyi divergence replaced by the Araki-Masuda divergence. 

To support our conjecture we derive a bound in one direction. Following the footsteps of~\cite{mosonyiogawa13}, it is easy to show using the multiplicativity of $L_p$-norms under tensor products\footnote{This property essentially follows from the characterization of vectors with definite $L_p$-norm in Lemma~\ref{lem:representation_Lp}.} and the data-processing inequality (Theorem~\ref{thm:data_processing}) that
\begin{align}
  D_{\alpha}(\rho\|\tau) &= \frac{1}{n} D_{\alpha}(\rho_n\|\tau_n) \label{eq:strong1} \\
   &\geq \frac{1}{n} D_{\alpha} \left( \left( \begin{matrix} \rho_n(T_n) & 0 \\ 0 & \rho_n(\id - T_n) \end{matrix} \right) \,\middle\|\, \left( \begin{matrix} \tau_n(T_n) & 0 \\ 0 & \tau_n(\id - T_n) \end{matrix} \right)\right) \label{eq:strong2} \\
   &\geq \frac{1}{n(\alpha - 1)} \log \rho_n(T_n)^{\alpha} \tau_n(T_n)^{1-\alpha} ,
\end{align}
for any test $T_n$ and any $\alpha > 1$. This implies after some manipulations that
\begin{align}
B_e^{*}(r) \geq \sup_{\alpha > 1} \frac{\alpha-1}{\alpha} \left( r - D_{\alpha}(\rho\|\tau) \right)\,.
\end{align}
Since this derivation holds for all R\'enyi divergences that satisfy additivity~\eqref{eq:strong1} and data-processing~\eqref{eq:strong2}, our conjecture would also imply that the Araki-Masuda divergences are minimal amongst all generalizations of R\'enyi divergences satisfying these two properties.


\section*{Acknowledgments}

We thank Anna Jen\u{c}ov\'a for pointing out a mistake in Eq.~\eqref{eq:pnorm2} in a previous version of this manuscript~\cite{jencova17}. MB and MT thank the Department of Physics at Ghent University, and MB and VBS thank the School of Physics at University of Sydney for their hospitality while part of this work was done.
VBS is supported by the EU through the ERC Qute. MT is funded by an ARC Discovery Early Career Researcher Award (DECRA) fellowship and acknowledges support from the ARC Centre of Excellence for Engineered Quantum Systems (EQUS). 


\appendix

\section{Complex powers of unbounded operators}\label{app:thelemma}

The purpose of this appendix is to collect some basics statements concerning matrix elements of complex powers of spatial derivative operators. In particular, we prove some basic analyticity properties needed for our proofs based on interpolation theory. We start with recalling a lemma about complex powers of unbounded operators, which can, e.g., be found in~\cite[Lem.~VI.2.3]{takesakibook2} and is restated here for the reader's convenience. In order to shorten the notation, we now let for the rest of the appendix the strip $S_{\gamma}$ be defined for real $\gamma >0$ as
\begin{align}
  S_{\gamma} = \{z \in \C \,:\,0 \leq \Re(z) \leq \gamma\}\,.
\end{align}

\begin{lemma}\label{lem:complexpower}
  Let $A$ be a densely define self-adjoint positive operator on a Hilbert space $\cH$, and let $\vec{\psi} \in \dom(A^{\gamma})$, the domain of $A^{\gamma}$. Then the $\cH$-valued function $z \in S_\gamma \mapsto A^z\vec{\psi}$ is continuous and bounded on the closure of $S_\gamma$ and holomorphic in the interior of $S_\gamma$.  
\end{lemma}

Although this lemma is stated for a densely defined operator $A$, it also holds for operators with non-trivial support projection $P_A$, if the domain $\dom(A)$ of $A$ is dense in $P_A\cH$, the support of $A$. This property will be exploited in the proof of the following lemma, which is a slightly generalized form of Lemma~\ref{lem:main_interpolation}. Its proof idea is similar to~\cite[Lem.~A]{araki82}.

\begin{lemma}\label{lem:main_interpolation_app}
  Let $A$ and $B$ be two self-adjoint positive operators on a Hilbert space $\cK$ and $\cH$ such that there domains are dense in their support, a bounded operator $V:\cH \to \cK$ be given. Consider two affine holomorphic functions $g(z) = g_1 z + g_0$, $h(z) = h_1 z + h_0$, with $g_0,g_1,h_0,h_1 \in \Rl$ and a vector $\vec{\vphi}\in \cH$, such that for $x \in \{0,1\}$ the statement $\vec{\vphi} \in \dom(B^{g(x)})$ holds. If the vector-valued function
  \begin{align}
    f: S_1 \to \cK,\quad z \mapsto A^{h(z)}\,V\,B^{g(z)}\vec{\vphi}
  \end{align}
  is uniformly bounded on the boundaries of the strip $S_1$, i.e.\ we have for $x \in \{0,1\}$ that 
  \begin{align}
    C_x := \norm{f(x+it)} < \infty
  \end{align}
  uniform in $t \in \Rl$, then $f(z)$ is holomorphic in the interior of $S_1$ and satisfies
  \begin{align}
    \norm{f(\theta)} \leq C_0^{1-\theta}\,C_1^{\theta}\,, \qquad \textrm{for $0 \leq \theta \leq 1$.}
  \end{align}
\end{lemma}

\begin{proof}
  We choose an element $\vec{\psi} \in \cH$ such that $\vec{\psi} \in \dom(A^{h(x)})$ for $x \in \{0,1\}$. Moreover, we have $\vec{\vphi} \in \dom(B^{g(x)})$, $x \in \{0,1\}$, from which we can infer by Lemma~\ref{lem:complexpower} that the $\cK$-valued functions $z\in S_1 \mapsto A^{h(z)}\vec{\psi}$ and $z \in S_1 \mapsto V B^{g(z)}\vec{\vphi}$ are holomorphic in the interior of $S_1$ and bounded on its closure. It follows by Hartog's theorem that the function
  \begin{align}
    \tilde{f}: z \in S_1 \mapsto \Scp{A^{h(\bar{z})}\vec{\psi}}{V\,B^{g(z)}\vec{\vphi}}
  \end{align}
  is also holomorphic in the interior and bounded continuous on the closure of $S_1$. We may then apply Hadamard's three-line theorem and arrive at
  \begin{align}
    |\sup_{t\in \Rl}\tilde{f}(\theta + it)| \leq \left(|\sup_{t\in \Rl}\tilde{f}(0+it)|\right)^{1-\theta} \, \left(|\sup_{t\in \Rl}\tilde{f}(1 + it)|\right)^{\theta}
  \end{align}
  for $\theta \in [0,1]$. The assumption $\norm{f(x+it)} < \infty$ implies that the vectors $B^{g(x+it)}\vec{\vphi}$ are actually elements of $\dom(A^{h(x+it)})$, again for $x\in \{0,1\}$, and hence
  \begin{align}
    \tilde{f}(x + it)&= \Scp{A^{h(x+it)}\vec{\psi}}{V B^{g(x+it)}\vec{\varphi}} = \Scp{\vec{\psi}}{A^{h(x+it)} V B^{g(x+it)}\vec{\varphi}} \leq \norm{\vec{\psi}} \,C_x\,.
  \end{align}
  This yields the upper bound
  \begin{align}
    |\tilde{f}(\theta + it)| \leq \norm{\vec{\psi}}\, C_0^{1-\theta}\,C_1^{\theta}\,,
  \end{align}
  for all $t \in \Rl$ and $0 \leq \theta \leq 1$. Due to the uniform bound in terms of the $\norm{\vec{\psi}}$ it then follows that $V\,B^{g(z)}\vec{\psi}$ is in the domain of $A^{h(z)}$ for all values $z \in S_1$ and hence we have
  \begin{align}
    \Scp{\vec{\psi}}{f(z)} = \tilde{f}(z)\,.
  \end{align}
  Hence the function $\Scp{\vec{\psi}}{f(z)}$ is weakly holomorphic and satisfies
  \begin{align}
    |\Scp{\vec{\psi}}{f(\theta + it)}| \leq \norm{\vec{\psi}} \, C_0^{1-\theta}\,C_1^{\theta}\,,
  \end{align}
   for all $t \in \Rl$ and $0 \leq \theta \leq 1$ and a dense set of vector $\vec{\psi}$. The assertion follows.
\end{proof}


\section{Properties of Petz' R\'enyi divergence}\label{app:misc_lemmas}

Here we present some properties of Petz' version of R\'enyi divergences, which proof is quite standard and thus only sketched. 

\begin{proof}[Proof of Proposition~\ref{pr:petzprop}]
  We start with (1), the statement on monotonicity. Applying Lemma~\ref{lem:main_interpolation} as in the proof of Lemma~\ref{lem:simpleinterpolate}, we find in an analogous way that the function $t \mapsto \log \|\spat{\rho}{\sigma}^{\frac{t}{2}}\vec{\rho}\|$ with $\vec{\rho} \in \cG_{\rho}$ is convex on $[-1,1]$ for any $\sigma \in \cP(\cM)$. We then proceed as in the proof of Lemma~\ref{lem:alpha_mono}.

  For (2) we note that by the definition of the map $R^\sigma(\vec{\rho})$ we have
  \begin{align}
    \widebar{D}_2(\rho\|\sigma)&=\log\langle\vec{\rho}|\spat{\rho}{\sigma}\vec{\rho}\rangle
    =\log\left\|R^{\sigma}(\vec{\rho})^{\dagger}(\vec{\rho})\right\|^2\notag\\
    &\leq\log\left\|R^{\sigma}(\vec{\rho})^{\dagger}\right\|^{2}=D_{\infty}(\rho\|\sigma)\,, 
  \end{align}
  and the assertion follows.

  Now consider the statements in (3). The limit from below is stated in~\cite[Sec.~7]{ohya93} (or alternatively see~\cite[Prop.~5.4(3)]{jaksic12}). For the limit from above, let $P(\mathrm{d}t)$ denote the spectral measure of the positive self-adjoint operator $\spat{\rho}{\sigma}$, and set $\mu(\mathrm{d}t) = \Scp{\vec{\rho}}{P(dt) \vec{\rho}}$. Since the Petz R\'enyi divergences are assumed to be finite in the open interval $(1,2)$, we may invoke L'Hospital's rule to calculate the right limit
  \begin{align}
   \lim_{\alpha \searrow 1}\widebar{D}_{\alpha}(\rho\|\sigma) = \lim_{\alpha \searrow 1}\frac{\log\left(\int_0^{\infty}t^{\alpha-1}\mu(\mathrm{d}t)\right)}{\alpha-1} = \int_0^{\infty} \log(t) \mu(\mathrm{d}t) = D(\rho\|\sigma) \,.
  \end{align}
\end{proof}

Next we provide the proof of the ALT inequality for the case $1\leq p \leq 2$, which is unfortunately less elegant than for $p\geq2$.

\begin{proof}[Proof of Theorem~\ref{thm:alt} for the case $1\leq p\leq2$.]
Let $\vec{\omega}\in\cH$ with $\|\vec{\omega} \| = 1$ such that $\omega\lll\sigma$ for $\omega\in\cS(\cM)$ implemented by $\vec{\omega}$. It follows that $\vec{\sigma}\in \dom\left(\spat{\omega}{\sigma}\right)$ and we also (trivially) have $\vec{\sigma}\in \dom\left(\spat{\rho}{\sigma}^{\frac{1}{2}}\right)$. By Hartog's theorem and Lemma~\ref{lem:complexpower} the function
\begin{align}
  f(z):=\Scp{\spat{\omega}{\sigma}^{1-\frac{\bar{z}}{p}}\vec{\sigma}}{\spat{\rho}{\sigma}^{\frac{z}{2}}\vec{\sigma}}
\end{align}
is holomorphic in the interior and bounded on the closure of $S_1$. We may then apply Hadamard's three-line theorem and arrive at
\begin{align}
\left|f\left(\frac{p}{2}\right)\right|\leq\sup_{t\in\Rl}|f(1+it)|^{\frac{p}{2}}\sup_{t\in\Rl}|f(it)|^{1-\frac{p}{2}}\,.
\end{align}
We estimate
\begin{align}
  |f(it)|&=\left|\Scp{\spat{\omega}{\sigma}^{1+\frac{it}{p}}\vec{\sigma}}{\spat{\rho}{\sigma}^{\frac{it}{2}}\vec{\sigma}}\right| = \left|\Scp{\spat{\omega}{\sigma}^{\half+\frac{it}{p}}\vec{\omega}}{\spat{\rho}{\sigma}^{\frac{it}{2}}\vec{\sigma}}\right|\\
         &=\left|\Scp{\spat{\omega}{\sigma}^{it\left(\frac{1}{p}+\frac{1}{2}\right)}\vec{\omega}}{\spat{\omega}{\sigma}^{\half+\frac{it}{2}}\spat{\rho}{\sigma}^{\frac{it}{2}}\vec{\sigma}}\right|\\
         &\leq \norm{\spat{\omega}{\sigma}^{\half+\frac{it}{2}}\spat{\rho}{\sigma}^{\frac{it}{2}}\vec{\sigma}}\,,
\end{align}
where we used multiple times that the imaginary power of the spatial derivative is an isometry. We may then invoke Lemma~\ref{lem:modular_identity} for $z=\frac{it}{2}$ and get that
\begin{align}
  |f(it)| \leq \norm{\spat{\rho}{\omega}^{\frac{it}{2}}\vec{\omega}} \leq 1\,.
\end{align}
We are left to estimate the value of the function on the line $1+it$. This yields
\begin{align}
  |f(1+it)| &= \left|\Scp{\spat{\omega}{\sigma}^{1-\frac{1}{p}-\frac{it}{p}}\vec{\sigma}}{\spat{\rho}{\sigma}^{\frac{it}{2}}\vec{\rho}}\right| \\
  &\leq \norm{\spat{\omega}{\sigma}^{\frac{1}{q}-\frac{it}{p}}\vec{\sigma}}_{q,\sigma} \, \norm{\spat{\rho}{\sigma}^{\frac{it}{2}}\vec{\rho}}_{p,\sigma}
\end{align}
by norm duality on the space $\cH_\sigma$ for $q$ determined by $\frac{1}{p}+\frac{1}{q} = 1$. Invoking Lemma~\ref{lem:representation_Lp} then yields a value of one for the first factor. This is almost what we aimed for, except of the appearance of $\spat{\rho}{\sigma}^{\frac{it}{2}}$ in the second factor. However, let us consider the state on $\cM$ induced by this vector:
\begin{align}
  a \in \cM \mapsto \Scp{\spat{\rho}{\sigma}^{\frac{it}{2}}\vec{\rho}}{\pi_\sigma(a)\spat{\rho}{\sigma}^{\frac{it}{2}}\vec{\rho}} = \Scp{\vec{\rho}}{\spat{\rho}{\sigma}^{\frac{-it}{2}}\pi_\sigma(a)\spat{\rho}{\sigma}^{\frac{it}{2}}\vec{\rho}}\,.
\end{align}
Standard facts from the modular theory of von Neumann algebras (see, e.g., \cite[Thm VIII.1.2 and Thm. IX.3.8]{takesakibook2}) now imply that the one-parameter group, called the modular group of the state $\rho$,
\begin{align}
  a \in \cM \mapsto d^{\rho}_t(a):= \spat{\rho}{\sigma}^{\frac{-it}{2}}\pi_\rho(a)\spat{\rho}{\sigma}^{\frac{it}{2}}
\end{align}
is independent of the state $\sigma$ and leaves the state $\rho$ in variant. That is, we have $\rho(d^\rho_t(a)) = \rho(a)$. It follows that the vector $\spat{\rho}{\sigma}^{\frac{it}{2}}\vec{\rho}$ is in fact implementing the state $\rho$ and by Lemma~\ref{lem:repindep} we arrive at the bound
\begin{align}
\left|\Scp{\vec{\omega}}{\spat{\rho}{\sigma}^{\frac{p}{4}}\vec{\sigma}}\right| = \left|\Scp{\spat{\omega}{\sigma}^{\half}\vec{\sigma}}{\spat{\rho}{\sigma}^{\frac{p}{4}}\vec{\sigma}}\right| \leq \norm{\rho}_{p,\sigma}^{\frac{p}{2}}\,.
\end{align}
The assertion follows by taking the supremum over vectors $\vec{\omega} \in \cH$ with $\norm{\vec{\omega}} =1$ such that $\omega \lll \sigma$, since vectors with this property are dense in $\cH_\sigma$.
\end{proof}


\section{Facts about $L_p$-spaces and spatial derivatives}

Here we collect a few lemmas about Araki-Masuda $L_p$-spaces, which were first proven in~\cite{araki82} and immediately generalize to our setting. We do not repeat the arguments here but point the reader to the relevant sections in Araki and Masuda's paper. We notify the reader that compared to~\cite{araki82}, the roles of $\cM$ and its commutant $\cM^\prime$ have to be switched, i.e.\ our spatial derivative equals the relative modular operator defined with respect to states on the commutant $\cM^\prime$ of $\cM$. We start by stating a lemma constructing a dense set of vectors with bounded Araki-Masuda $L_p$-norm. 

\begin{lemma}\label{lem:representation_Lp}
Let $\vec{\sigma}\in\cH$ be a vector implementing $\sigma\in\cP(\cM)$. Furthermore, let $\vec{\omega}\in\cH$ and let $\cM \subset \cB(\cH)$. Then, we have
\begin{align}
\left\|u\spat{\omega}{\sigma}^{\frac{1}{p}+it}\vec{\sigma}\right\|_{p,\sigma}\leq\|\vec{\omega}\|^{\frac{2}{p}}\,,
\end{align}
for $1\leq p\leq\infty$ and all $u\in\cM^\prime$ with $\|u\|\leq1$. Moreover, for $2 \leq p \leq \infty$ there exist for $\vec{\rho} \in \cH$ with $\norm{\vec{\rho}}_{p,\sigma} < \infty$ a unique vector $\vec{\omega}\in\cH$ as well as a partial isometry $u \in \cM^\prime$ such that 
\begin{align}
  \vec{\rho} = u\spat{\omega}{\sigma}^{\frac{1}{p}}\vec{\sigma}\quad\text{as well as}\quad\norm{\vec{\rho}}_{p,\sigma} = \|\vec{\omega}\|^{\frac{2}{p}}\,.
\end{align}
\end{lemma}

\begin{proof}
  This follows for $1\leq p< \infty$ as in~\cite[Lem 4.1]{araki82}. The case $p=\infty$ is trivial.
\end{proof}

The following Lemma provides some simple estimates on the value of the $L_p$-norms.

\begin{lemma}\label{lem:lem6.1_araki}
Let $\rho, \sigma \in \cP(\cM)$. For $2\leq p\leq\infty$ we have $\|\rho\|_{p,\sigma}\geq\|\vec{\rho}\|\|\vec{\sigma}\|^{\frac{2}{p}-1}$ and for $1\leq p\leq2$ the inequality reverses.
\end{lemma}

\begin{proof}
This follows as in~\cite[Lem.~6.1]{araki82}.
\end{proof}

The final lemma is often of use in estimating the $L_p$-norm of a vector. 

\begin{lemma}\label{lem:modular_identity}
Let $\rho,\sigma \in\cP(\cM)$ and let $\cH$ be a Hilbert space supporting vectors $\vec{\rho}$ and $\vec{\sigma}$ implementing $\rho$ and $\sigma$, respectively. Then, we have for all $\vec{\omega}\in\cH$ that
\begin{align}
\left\|\spat{\omega}{\sigma}^{\frac{1}{2}-z}\spat{\rho}{\sigma}^{z}\vec{\sigma}\right\|=\left\|\spat{\rho}{\omega}^{z}\vec{\omega}\right\|\,,
\end{align}
with $0\leq\Re(z)\leq\frac{1}{2}$, and where $\omega$ is the positive functional on $\cM$ implemented by $\vec{\omega}$.
\end{lemma}

\begin{proof}
This follows as in~\cite[Lem.~C.2]{araki82}.
\end{proof}


\section{Operator Jensen for unbounded operators}
\label{app:hansen}

\begin{lemma}
  Let $A$ be a closed positive operator on a Hilbert space $\cH$, $\psi \in \cH$, and $V$ a contraction on $\cH$ such that $V\psi \in D(A)$. Then, for $0<t<1$, we have that $D((V^* A V)^{t/2}) \subset D((V^* A^{t} V)^{\half})$ and
  \begin{align}
    \Scp{\psi}{V^* A^{t} V \psi} \leq \Scp{\psi}{(V^* A V)^{t} \psi} \,.
  \end{align}
\end{lemma}

\begin{proof}
  We first note that if $V\psi \in D(A)$ then also $V\psi \in D(A^t)$ for $0<t<1$.
  Let $E_\gamma = \int_{-\gamma}^\gamma P(\mathrm{d}t)$ be a spectral projection of $A = \int_{-\infty}^\infty t P(\mathrm{d}t)$ such that $E_{\gamma}AE_{\gamma} = AE_{\gamma}$ is bounded. Hence it holds that
  \begin{align}
    \label{eq:jensenopeq1}    
    \Scp{V\psi}{(AE_{\gamma})^{t} V \psi} \leq \Scp{\psi}{(V^* A^\half E_{\gamma} A^\half V)^{t} \psi} = \norm{(V^* A^\half E_{\gamma} A^\half V)^{t/2} \psi}
  \end{align}
  Consider now the two operators $V^* A^\half E_{\gamma} A^\half V$ and $V^* A V$. Since $E_{\gamma} \leq \id$, we have
  \begin{align}
    \norm{E_{\gamma} A^\half V \xi} \leq \norm{A^\half V \xi}\,,
  \end{align}
  for all $\xi \in D(A^\half V)$. Moreover, we naturally have $D(A^\half V) \subset D(E_{\gamma} A^\half V)$. We then apply~\cite[Lemma D]{araki82} to get
  \begin{align}
    \label{eq:jensenopeq2}
    \norm{(V^* A^\half E_{\gamma} A^\half V)^{t/2} \psi} \leq \norm{(V^* A^\half A^\half V)^{t/2} \psi} = \norm{(V^* A V)^{t/2} \psi}\,.
  \end{align}
  Since spectral projections commute with applying functions to an operator, we also have that
  \begin{align}
    \label{eq:jensenopeq3}
    \lim_{\gamma \to 1} \Scp{V\psi}{(AE_\gamma)^{t} V \psi} = \lim_{\gamma \to 1} \Scp{E_{\gamma}V\psi}{(A)^{t} E_{\gamma} V \psi} = \Scp{V\psi}{(A)^{t} V \psi}\,.
  \end{align}
   Inserting now Eqs.~\eqref{eq:jensenopeq2} and~\eqref{eq:jensenopeq3} in Eq.~\eqref{eq:jensenopeq1} proves the assertion.
\end{proof}


\bibliographystyle{alphaarxiv}
\bibliography{library}

\newcommand{\etalchar}[1]{$^{#1}$}
\begin{thebibliography}{MLDS{\etalchar{+}}13}

\bibitem[AD15]{audenaert13}
K.~M.~R. Audenaert and N.~Datta.
\newblock {$\alpha$-z-Relative Renyi Entropies}.
\newblock \href{http://dx.doi.org/10.1063/1.4906367}{{\em Journal of
  Mathematical Physics} {\bf 56}:\,022202} (2015).

\bibitem[Alb83]{alberti83}
P.~M. Alberti.
\newblock {A Note on the Transition Probability over C*-Algebras}.
\newblock \href{http://dx.doi.org/10.1007/BF00398708}{{\em Letters in
  Mathematical Physics} {\bf 7}:\,25--32} (1983).

\bibitem[AM82]{araki82}
H.~Araki and T.~Masuda.
\newblock {Positive Cones and Lp-Spaces for von Neumann Algebras}.
\newblock {\em Publications of the Research Institute for Mathematical
  Sciences, Kyoto University} {\bf 18}:\,339--411, (1982).

\bibitem[Ara76]{araki76}
B.~H. Araki.
\newblock {Relative Entropy of States of von Neumann Algebras}.
\newblock {\em Publications of the Research Institute for Mathematical
  Sciences, Kyoto University} {\bf 11}:\,809--833, (1976).

\bibitem[Ara90]{araki90}
H.~Araki.
\newblock {On an Inequality of Lieb and Thirring}.
\newblock \href{http://dx.doi.org/10.1007/BF01045887}{{\em Letters in
  Mathematical Physics} {\bf 19}(2):\,167--170} (1990).

\bibitem[Bei13]{beigi13}
S.~Beigi.
\newblock {Sandwiched R{\'{e}}nyi Divergence Satisfies Data Processing
  Inequality}.
\newblock \href{http://dx.doi.org/10.1063/1.4838855}{{\em Journal of
  Mathematical Physics} {\bf 54}(12):\,122202} (2013).

\bibitem[BFS16]{furrer11}
M.~Berta, F.~Furrer, and V.~B. Scholz.
\newblock {The Smooth Entropy Formalism for von Neumann Algebras}.
\newblock \href{http://dx.doi.org/10.1063/1.4936405}{{\em Journal of
  Mathematical Physics} {\bf 57}(1):\,015213} (2016).

\bibitem[BSW15]{bertawilde14}
M.~Berta, K.~Seshadreesan, and M.~Wilde.
\newblock {R{\'{e}}nyi Generalizations of the Conditional Quantum Mutual
  Information}.
\newblock \href{http://dx.doi.org/10.1063/1.4908102}{{\em Journal of
  Mathematical Physics} {\bf 56}(2):\,022205} (2015).

\bibitem[BW75]{Bisognano_1975}
J.~J. Bisognano and E.~H. Wichmann.
\newblock On the duality condition for a Hermitian scalar field.
\newblock \href{http://dx.doi.org/10.1063/1.522605}{{\em Journal of
  Mathematical Physics} {\bf 16}(4):\,985--1007} (1975).

\bibitem[Dat09]{datta08}
N.~Datta.
\newblock {Min- and Max- Relative Entropies and a New Entanglement Monotone}.
\newblock \href{http://dx.doi.org/10.1109/TIT.2009.2018325}{{\em IEEE
  Transactions on Information Theory} {\bf 55}(6):\,2816--2826} (2009).

\bibitem[FL13]{frank13}
R.~L. Frank and E.~H. Lieb.
\newblock {Monotonicity of a Relative R{\'{e}}nyi Entropy}.
\newblock \href{http://dx.doi.org/10.1063/1.4838835}{{\em Journal of
  Mathematical Physics} {\bf 54}(12):\,122201} (2013).

\bibitem[Haa79]{haagerup79}
U.~Haagerup.
\newblock {Lp-Spaces Associated with an Arbitrary von Neumann Algebra}.
\newblock In {\em Proc. Colloquium Marseille 1977}, Volume 274 of Algebres
  d'Op{\'{e}}rateurs et leurs Applications en Physique Math{\'{e}}matique,
  pages 175----184, (1979).

\bibitem[HJX09]{haagerup10}
U.~Haagerup, M.~Junge, and Q.~Xu.
\newblock {A Reduction Method for Noncommutative Lp-spaces and Applications}.
\newblock \href{http://dx.doi.org/10.1090/S0002-9947-09-04935-6}{{\em
  Transactions of the American Mathematical Society} {\bf
  362}(04):\,2125--2165} (2009).

\bibitem[HP03]{hansen03}
F.~Hansen and G.~K. Pedersen.
\newblock {Jensen's Operator Inequality}.
\newblock \href{http://dx.doi.org/10.1112/S0024609303002200}{{\em Bulletin
  London Mathematical Society} {\bf 35}(4):\,553--564} (2003).

\bibitem[Jen16]{jencova16}
A.~Jen\u{c}ov\'{a}.
\newblock {R\'{e}nyi relative entropies and noncommutative $L_p$-spaces}.
\newblock {\em Preprint }\href{http://arxiv.org/abs/1609.08462}{arXiv:
  1609.08462} (2016).

\bibitem[Jen17]{jencova17}
A.~Jen\u{c}ov\'{a}.
\newblock {R\'{e}nyi relative entropies and noncommutative $L_p$-spaces II}.
\newblock {\em Preprint }\href{http://arxiv.org/abs/1707.00047}{arXiv:
  1707.00047} (2017).

\bibitem[JOPP12]{jaksic10}
V.~Jaksic, Y.~Ogata, Y.~Pautrat, and C.~A. Pillet.
\newblock {Entropic Fluctuations in Quantum Statistical Mechanics --- An
  Introduction}.
\newblock In {\em Quantum Theory from Small to Large Scales: Lecture Notes of
  the Les Houches Summer School}, Volume~95. Oxford University Press (2012).

\bibitem[JOPS12]{jaksic12}
V.~Jak{\v{s}}i{\'{c}}, Y.~Ogata, C.-A. Pillet, and R.~Seiringer.
\newblock {Quantum Hypothesis Testing and Non-Equilibrium Statistical
  Mechanics}.
\newblock \href{http://dx.doi.org/10.1142/S0129055X12300026}{{\em Reviews in
  Mathematical Physics} {\bf 24}(06):\,1230002} (2012).

\bibitem[JRS02]{jain02}
R.~Jain, J.~Radhakrishnan, and P.~Sen.
\newblock {Privacy and Interaction in Quantum Communication Complexity and a
  Theorem About the Relative Entropy of Quantum States}.
\newblock In {\em The 43rd Annual IEEE Symposium on Foundations of Computer
  Science, 2002. Proceedings.}, pages 429--438, Vancouver, (2002). IEEE
  Computer Society.

\bibitem[Kos84]{Kosaki_1984}
H.~Kosaki.
\newblock Applications of the complex interpolation method to a von Neumann
  algebra: Non-commutative Lp-spaces.
\newblock \href{http://dx.doi.org/10.1016/0022-1236(84)90025-9}{{\em Journal of
  Functional Analysis} {\bf 56}(1):\,29--78} (1984).

\bibitem[Kos92]{kosaki92}
H.~Kosaki.
\newblock {An Inequality of Araki-Lieb-Thirring (Von Neumann Algebra Case)}.
\newblock \href{http://dx.doi.org/10.2307/2159671}{{\em Proceedings of the
  American Mathematical Society} {\bf 114}(2):\,477} (1992).

\bibitem[LR73]{lieb73}
E.~H. Lieb and M.~B. Ruskai.
\newblock {Proof of the Strong Subadditivity of Quantum-Mechanical Entropy}.
\newblock \href{http://dx.doi.org/10.1063/1.1666274}{{\em Journal of
  Mathematical Physics} {\bf 14}(12):\,1938} (1973).

\bibitem[LT76]{lieb76}
E.~H. Lieb and W.~E. Thirring.
\newblock {\em {Inequalities for the Moments of the Eigenvalues of the
  Schr{\"{o}}dinger Hamiltonian and Their Relation to Sobolev Inequalities}}.
\newblock Princeton University Press (1976).

\bibitem[MHR17]{hermes15}
A.~M{\"u}ller-Hermes and D.~Reeb.
\newblock Monotonicity of the Quantum Relative Entropy Under Positive Maps.
\newblock \href{http://dx.doi.org/10.1007/s00023-017-0550-9}{{\em Annales Henri
  Poincar{\'e}} {\bf 18}(5):\,1777--1788} (2017).

\bibitem[MLDS{\etalchar{+}}13]{lennert13}
M.~M{\"{u}}ller-Lennert, F.~Dupuis, O.~Szehr, S.~Fehr, and M.~Tomamichel.
\newblock {On Quantum R{\'{e}}nyi Entropies: A New Generalization and Some
  Properties}.
\newblock \href{http://dx.doi.org/10.1063/1.4838856}{{\em Journal of
  Mathematical Physics} {\bf 54}(12):\,122203} (2013).

\bibitem[MO15]{mosonyiogawa13}
M.~Mosonyi and T.~Ogawa.
\newblock {Quantum Hypothesis Testing and the Operational Interpretation of the
  Quantum R{\'{e}}nyi Relative Entropies}.
\newblock \href{http://dx.doi.org/10.1007/s00220-014-2248-x}{{\em
  Communications in Mathematical Physics} {\bf 334}(3):\,1617--1648} (2015).

\bibitem[MO17]{mosonyi14-2}
M.~Mosonyi and T.~Ogawa.
\newblock Strong Converse Exponent for Classical-Quantum Channel Coding.
\newblock \href{http://dx.doi.org/10.1007/s00220-017-2928-4}{{\em
  Communications in Mathematical Physics} {\bf 355}(1):\,373--426} (2017).

\bibitem[ON00]{ogawa00}
T.~Ogawa and H.~Nagaoka.
\newblock {Strong Converse and Stein's Lemma in Quantum Hypothesis Testing}.
\newblock \href{http://dx.doi.org/10.1109/18.887855}{{\em IEEE Transactions on
  Information Theory} {\bf 46}(7):\,2428--2433} (2000).

\bibitem[OP93]{ohya93}
M.~Ohya and D.~Petz.
\newblock {\em {Quantum Entropy and Its Use}}.
\newblock Springer (1993).

\bibitem[Pet85]{petz85}
D.~Petz.
\newblock {Quasientropies for States of a von Neumann Algebra}.
\newblock \href{http://dx.doi.org/10.2977/prims/1195178929}{{\em Publications
  of the Research Institute for Mathematical Sciences} {\bf 21}(4):\,787--800}
  (1985).

\bibitem[Pet86]{petz86}
D.~Petz.
\newblock {Quasi-entropies for Finite Quantum Systems}.
\newblock \href{http://dx.doi.org/10.1016/0034-4877(86)90067-4}{{\em Reports on
  Mathematical Physics} {\bf 23}(1):\,57--65} (1986).

\bibitem[PX03]{pisier03}
G.~Pisier and Q.~Xu.
\newblock {Non-Commutative Lp-Spaces}.
\newblock In {\em Handbook of the Geometry of Banach Spaces}, Volume~2,
  chapter~34, pages 1459--1517.
  \href{http://dx.doi.org/10.1016/S1874-5849(03)80041-4}{Elsevier} (2003).

\bibitem[Sti55]{stinespring54}
W.~F. Stinespring.
\newblock {Positive Functions On C*-Algebras}.
\newblock \href{http://dx.doi.org/10.1090/S0002-9939-1955-0069403-4}{{\em
  Proceedings of the Americal Mathematical Society} {\bf 6}:\,211--216} (1955).

\bibitem[Tak79]{takesakibook1}
M.~Takesaki.
\newblock {\em {Theory of Operator Algebras I}}.
\newblock \href{http://dx.doi.org/10.1007/978-1-4612-6188-9}{Springer New York}
  (1979).

\bibitem[Tak03]{takesakibook2}
M.~Takesaki.
\newblock {\em {Theory of Operator Algebras II}}.
\newblock Volume 125 of Encyclopaedia of Mathematical Sciences,
  \href{http://dx.doi.org/10.1007/978-3-662-10451-4}{Springer Berlin
  Heidelberg} (2003).

\bibitem[Tom12]{mytutorial12}
M.~Tomamichel.
\newblock {Smooth entropies: A Tutorial With Focus on Applications in
  Cryptography}, (2012).
\newblock \\ Online:
  \url{http://www.marcotom.info/talks/QCrypt2012Tutorial.pdf}.

\bibitem[Tom16]{mybook}
M.~Tomamichel.
\newblock {\em {Quantum Information Processing with Finite Resources ---
  Mathematical Foundations}}.
\newblock Volume~5 of SpringerBriefs in Mathematical Physics,
  \href{http://dx.doi.org/10.1007/978-3-319-21891-5}{Springer International
  Publishing} (2016).

\bibitem[TWW17]{tomamichelww14}
M.~Tomamichel, M.~M. Wilde, and A.~Winter.
\newblock {Strong Converse Rates for Quantum Communication}.
\newblock \href{http://dx.doi.org/10.1109/TIT.2016.2615847}{{\em IEEE
  Transactions on Information Theory} {\bf 63}(1):\,715--727} (2017).

\bibitem[Uhl76]{Uhlmann76}
A.~Uhlmann.
\newblock The Transition Probability in the State Space of a *-algebra.
\newblock \href{http://dx.doi.org/10.1016/0034-4877(76)90060-4}{{\em Reports on
  Mathematical Physics} {\bf 9}(2):\,273--279} (1976).

\bibitem[Uhl77]{uhlmann77}
A.~Uhlmann.
\newblock {Relative entropy and the Wigner-Yanase-Dyson-Lieb concavity in an
  interpolation theory}.
\newblock \href{http://dx.doi.org/10.1007/BF01609834}{{\em Communications in
  Mathematical Physics} {\bf 54}(1):\,21--32} (1977).

\bibitem[Ume62]{umegaki62}
H.~Umegaki.
\newblock {Conditional Expectation in an Operator Algebra}.
\newblock \href{http://dx.doi.org/10.2996/kmj/1138844604}{{\em Kodai
  Mathematical Seminar Reports} {\bf 14}:\,59--85} (1962).

\bibitem[WWY14]{wilde13}
M.~M. Wilde, A.~Winter, and D.~Yang.
\newblock {Strong Converse for the Classical Capacity of Entanglement-Breaking
  and Hadamard Channels via a Sandwiched R{\'{e}}nyi Relative Entropy}.
\newblock \href{http://dx.doi.org/10.1007/s00220-014-2122-x}{{\em
  Communications in Mathematical Physics} {\bf 331}(2):\,593--622} (2014).

\end{thebibliography}

\end{document}